\documentclass[a4paper,10pt]{article}

\usepackage{mathrsfs}
\usepackage{a4wide}
\usepackage[utf8]{inputenc}
\usepackage[T1]{fontenc}
\usepackage{mathrsfs}
\usepackage{amsmath}
\usepackage{amssymb}
\usepackage{bussproofs}
\usepackage{thrmappendix}
\usepackage{color}
\usepackage{url}
\usepackage{stmaryrd}
\usepackage{tikz}
\usetikzlibrary{shapes}
\usetikzlibrary{arrows}

\newenvironment{varitemize}
{
\begin{list}{\labelitemi}
{\setlength{\itemsep}{0pt}
 \setlength{\topsep}{0pt}
 \setlength{\parsep}{0pt}
 \setlength{\partopsep}{0pt}
 \setlength{\leftmargin}{15pt}
 \setlength{\rightmargin}{0pt}
 \setlength{\itemindent}{0pt}
 \setlength{\labelsep}{5pt}
 \setlength{\labelwidth}{10pt}
}}
{
 \end{list}
}

\newcounter{numberone}

\newenvironment{varenumerate}
{
\begin{list}{\arabic{numberone}.}
{
  \usecounter{numberone}
  \setlength{\itemsep}{0pt}
  \setlength{\topsep}{0pt}
  \setlength{\parsep}{0pt}
  \setlength{\partopsep}{0pt}
  \setlength{\leftmargin}{15pt}
  \setlength{\rightmargin}{0pt}
  \setlength{\itemindent}{0pt}
  \setlength{\labelsep}{5pt}
  \setlength{\labelwidth}{15pt}
}}
{
\end{list}
} 

\newcommand{\bnf}{\;::=\;}
 
\newcommand{\proves}{\vdash}
\newcommand{\midd}{\; \; \mbox{\Large{$\mid$}}\;\;} 

\newcommand{\NN}{\mathbb{N}}
\newcommand{\BB}{\mathbb{B}}
\newcommand{\RR}{\mathbb{R}}

\newcommand{\values}{\ensuremath{\mathcal{V}}}

\newcommand{\dom}[1]{\,\mathit{dom}(#1)\,}
\newcommand{\ttrue}{\,\mathsf{true}\,}
\newcommand{\ffalse}{\,\mathsf{false}\,}

\newcommand{\nil}{\,\mathsf{nil}\,}

\newcommand{\evav}[1]{#1 \Downarrow}

\newcommand{\evalv}[2]{#1 \Downarrow \distr{#2}}

\newcommand{\evalvs}[2]{#1 \Rightarrow {\distr{#2}}}
\newcommand{\evalvsD}[2]{#1 \Rightarrow {{#2}}}
\newcommand{\evalvsp}[2]{#1 \Rightarrow^{#2}}
\newcommand{\distrv}[1]{{\{{#1}^1\}}}

\newcommand{\distr}[1]{\mathscr #1}
\newcommand{\diS}[1]{\llbracket #1\rrbracket}

\newcommand{\op}[2]{#1 \, \mathsf{op}\, #2}
\newcommand{\opn}[2]{#1 \, \tilde {\mathsf{op}}\, #2}

\newcommand{\fst}[1]{\mathsf{fst}\, (#1)}
\newcommand{\snd}[1]{\mathsf{snd}\, (#1)}
\newcommand{\abstrac}{\lambda x.}
\newcommand{\abstr}[2]{\lambda #1.#2}
\newcommand{\ifth}[3]{\mathsf{if}\, #1\,  \mathsf{then} \,#2 \,\mathsf{else}\, #3}
\newcommand{\fix}[1]{\mathsf{fix}\, x. \, #1}
\newcommand{\pair}[2]{\langle #1,#2 \rangle}
\newcommand{\caset}[3]{\mathsf{case}\, #1\, \mathsf{of}\, \{ \mathsf{nil}\rightarrow #2\;|\;h::t \rightarrow #3\}}
\newcommand{\lst}[2]{#1 :: #2}

\newcommand{\id}{I}

\newcommand{\termty}[1]{\mathcal{T}^{#1}}
\newcommand{\val}[1]{\mathcal{V}^{#1}}
\newcommand{\termtyg}[2]{\mathcal{T}_{#1}^{#2}}
\newcommand{\subst}[3]{#1[#2/#3]}
\newcommand{\R}{{\precsim^n_\circ}}

\newcommand{\Rh}{{{\precsim^n_\circ}^H}}

\newcommand{\Rv}{{\precsim_\circ}}
\newcommand{\Ev}{{\backsim_\circ}}
\newcommand{\Rvh}{\precsim_\circ^H}

\newcommand{\LOP}{\ensuremath{\Lambda_\oplus}}
\newcommand{\PCF}{\ensuremath{\mathsf{PCF}}}
\newcommand{\PCFL}{\ensuremath{\mathsf{PCFL}}}
\newcommand{\PCFLP}{\ensuremath{\mathsf{PCFL}_\oplus}}

\newcommand{\vars}{\mathcal{V}}
\newcommand{\ops}{\mathcal{O}}

\newcommand{\termone}{M}
\newcommand{\termtwo}{N}
\newcommand{\termthree}{L}
\newcommand{\termfour}{P}
\newcommand{\termfive}{R}
\newcommand{\termsix}{T}
\newcommand{\termseven}{U}

\newcommand{\varone}{x}
\newcommand{\vartwo}{y}
\newcommand{\varthree}{z}
\newcommand{\varfour}{w}

\newcommand{\valone}{V}
\newcommand{\valtwo}{W}
\newcommand{\valthree}{Z}

\newcommand{\emb}[1]{\langle\!\langle #1\rangle\!\rangle}
\newcommand{\embsn}[1]{\lceil #1\rceil}
\newcommand{\abstrsv}[1]{\lambda.#1}
\newcommand{\dmyval}{\star}
\newcommand{\fixt}{\mathsf{fix}}

\newcommand{\typtwo}{\tau}
\newcommand{\typthree}{\theta}

\newcommand{\mcone}{\mathcal{M}}

\newcommand{\stateone}{s}
\newcommand{\statetwo}{t}
\newcommand{\statethree}{r}

\newcommand{\bty}{\mathbf{bool}}
\newcommand{\ity}{\mathbf{int}}
\newcommand{\arr}[2]{#1\rightarrow #2}
\newcommand{\cprod}[2]{#1\times #2}
\newcommand{\lstty}[1]{[#1 ]}

\newcommand{\opc}{\mathsf{op}}
\newcommand{\resty}[1]{{#1}_{\mathsf{op}}}

\newcommand{\FV}[1]{\mathit{FV}(#1)}

\newcommand{\cnst}[1]{\underline{#1}}

\newcommand{\relone}{\mathcal{R}}
\newcommand{\reltwo}{\mathcal{P}}

\newcommand{\relonea}[2]{#1\;\relone\;#2}
\newcommand{\rela}[3]{#1\;#2\;#3}

\newcommand{\relct}{\mathscr{R}}

\newcommand{\vrelone}{R}
\newcommand{\vreltwo}{P}

\newcommand{\probone}{p}
\newcommand{\probtwo}{q}

\newcommand{\emptycon}{[\cdot]}
\newcommand{\ctxone}{C}
\newcommand{\ctxtwo}{D}

\newcommand{\ectxone}{E}

\newcommand{\act}[2]{#1[#2]}
\newcommand{\tycon}[4]{#1(#2;#3):#4}

\newcommand{\distrone}{\mathscr{D}}
\newcommand{\distrtwo}{\mathscr{E}}
\newcommand{\distrthree}{\mathscr{F}}

\newcommand{\supp}[1]{\mathsf{S}(#1)}

\newcommand{\nilact}{\mathit{nil}} 
\newcommand{\hd}{\mathit{hd}}
\newcommand{\tl}{\mathit{tl}}
\newcommand{\evact}{\mathit{eval}}
\newcommand{\fstact}{\mathit{fst}}
\newcommand{\sndact}{\mathit{snd}}

\newcommand{\mccbv}{\mathcal{M}_\oplus}
\newcommand{\mccbvstates}{\mathcal{S}_\oplus}
\newcommand{\mccbvlabels}{\mathcal{L}_\oplus}
\newcommand{\mccbvmatrix}{\mathcal{P}_\oplus}

\newcommand{\pasv}[1]{\ensuremath{\precsim_{#1}}}
\newcommand{\pabv}[1]{\ensuremath{\backsim_{#1}}}

\newcommand{\ngc}[1]{\mathit{CC}_{#1}}

\newcommand{\ccone}{\xi}

\newcommand{\types}{\mathcal{Y}}

\newcommand{\cqed}{}

\newcommand{\SSred}{\rightarrow}

\newcommand{\statone}{\mathcal{X}}
\newcommand{\fieldone}{\Sigma}
\newcommand{\actset}{\mathscr{A}}
\newcommand{\setone}{X}
\newcommand{\mpone}{\mathcal{C}}
\newcommand{\matrixone}{\mathcal{P}}
\newcommand{\powset}[1]{\mathscr{P}(#1)}
\newcommand{\stone}{s}
\newcommand{\lmp}[1]{\mathcal{C}_{#1}}

\newcommand{\testlan}{\mathscr{T}}
\newcommand{\testone}{t}
\newcommand{\testtwo}{s}
\newcommand{\actone}{a}
\newcommand{\seq}[1]{\langle #1\rangle}
\newcommand{\prsucc}[3]{P_{#1}(#2,#3)}

\newtheorem{lemma}{Lemma}
\newtheorem{definition}{Definition}
\newtheorem{proposition}{Proposition}
\newtheorem{theorem}{Theorem}
\newtheorem{example}{Example}
\newenvironment{proof}{\begin{trivlist}
       \item[\hskip \labelsep {\bfseries Proof.}]}{\hfill $\Box$ \end{trivlist}}

\begin{document}

\title{On Probabilistic Applicative Bisimulation\\ and Call-by-Value $\lambda$-Calculi (Long Version)}
\author{Rapha\"elle Crubill\'e\footnote{ENS-Lyon, \texttt{raphaelle.crubille@ens-lyon.fr}} \and 
        Ugo Dal Lago\footnote{Universit\`a di Bologna \& INRIA, \texttt{dallago@cs.unibo.it}}}
\date{}

\maketitle

\begin{abstract}
Probabilistic applicative bisimulation is a recently introduced coinductive methodology for program 
equivalence in a probabilistic, higher-order, setting. In this paper, the technique is generalized to 
a typed, call-by-value, lambda-calculus. Surprisingly, the obtained relation coincides with context 
equivalence, contrary to what happens when call-by-name evaluation is considered. Even more surprisingly, 
full-abstraction only holds in a symmetric setting. 
\end{abstract}

\section{Introduction}
Traditionally, an algorithm is nothing but a finite description of a sequence of deterministic primitive
instructions, which solve a computational problem when executed. Along the years, however, this concept has been generalized 
so as to reflect a broader class of effective procedures and machines. One of the many ways this has been done consists
in allowing probabilistic choice as a primitive instruction in algorithms, this way shifting from 
usual, deterministic computation to a new paradigm, called probabilistic computation. Examples of application 
areas in which probabilistic computation has proved to be useful include natural language processing~\cite{manning1999foundations}, 
robotics~\cite{thrun2002robotic}, computer vision~\cite{comaniciu2003kernel}, and machine learning~\cite{pearl1988probabilistic}.
Sometimes, being able to ``flip a fair coin'' while computing is a \emph{necessity} rather than an alternative, 
like in computational cryptography (where, e.g., secure public key encryption schemes must be
probabilistic~\cite{GoldwasserMicali}). 

Any (probabilistic) algorithm can be executed by concrete machines only once it takes the form of a \emph{program}. 
And indeed, various probabilistic programming languages have been introduced in the last years, 
from abstract ones~\cite{Plotkin,RamseyPfeffer,ParkPfenningThrun} to more concrete ones~\cite{Pfeffer01,Goodman}.
A quite common scheme consists in endowing any deterministic language with one or more primitives for probabilistic 
choice, like binary probabilistic choice or primitives for distributions.

Viewing algorithms as functions allows a smooth integration of distributions into the playground, itself
nicely reflected at the level of types through monads~\cite{GordonABCGNRR13,RamseyPfeffer}. As a matter of fact, some existing probabilistic
programming languages~\cite{Pfeffer01,Goodman} are designed around the $\lambda$-calculus or one of its incarnations, like
\textsf{Scheme}. This, in turn has stimulated foundational research about probabilistic $\lambda$-calculi, and in
particular about the nature of program equivalence in a probabilistic setting. 
This has already started to produce some interesting results in the realm of denotational semantics, where adequacy
and full-abstraction results have recently appeared~\cite{DanosHarmer,EPT13}.

Not much is known about operational techniques for probabilistic program equivalence, and in particular about coinductive
methodologies. This is in contrast with what happens for deterministic or nondeterministic programs, when various notions
of bisimulation have been introduced and proved to be adequate and, in some cases, fully abstract. A recent paper
by Alberti, Sangiorgi and the second author~\cite{ADLS13} generalizes Abramsky's applicative 
bisimulation~\cite{Abramsky90} to $\LOP$, a call-by-name, untyped $\lambda$-calculus endowed with binary, 
fair, probabilistic choice~\cite{DalLagoZorzi}. Probabilistic applicative bisimulation is shown to be
a congruence, thus included in context equivalence. Completeness, however, fails, the counterexample being exactly the
one separating bisimulation and context equivalence in a nondeterministic setting. Full abstraction is then recovered
when pure, deterministic $\lambda$-terms are considered, as well as well another, more involved, notion of
bisimulation, called coupled logical bisimulation, takes the place of applicative bisimulation.

In this paper, we proceed with the study of probabilistic applicative bisimulation, analysing its behaviour when instantiated
on call-by-value $\lambda$-calculi. This investigation brings up some nice, unexpected results. Indeed, not only
the non-trivial proof of congruence for applicative bisimulation can be adapted to the call-by-value setting, which is somehow
expected, but applicative bisimilarity turns out to precisely characterize context equivalence. This is quite surprising,
given that in nondeterministic $\lambda$-calculi, both when call-by-name and call-by-value evaluation are considered,
applicative bisimilarity is a congruence, but \emph{finer} than context equivalence. There is another, even less expected result:
the aforementioned correspondence does not hold anymore if we consider applicative \emph{simulation} and the
contextual \emph{preorder}.

Technically, the presented results owe much to a recent series of studies about probabilistic bisimulation for labelled
Markov processes~\cite{DEP02,vBMOW05}, i.e., labelled probabilistic transition systems in which the state space is continuous
(rather than discrete, as in Larsen and Skou's labelled Markov chains~\cite{LS91}), but time stays discrete. More specifically,
the way we prove that context equivalent terms are bisimilar goes by constructively show how 
each \emph{test} of a kind characterizing probabilistic bisimulation can be turned into an equivalent \emph{context}. 
If, as a consequence, two terms are not bisimilar, then any test the two terms satisfy with different probabilities 
(of which there must be at least one) becomes a context in which the two terms converges with different probabilities. 
This also helps understanding the discrepancies between the probabilistic and nondeterministic settings, since 
in the latter the class of tests characterizing applicative bisimulation is well-known to be quite large~\cite{Ong93}.

The whole development is done in a probabilistic variation on \PCF\ with lazy lists, called \PCFLP: working
on an applied calculus allows to stay closer to concrete programming languages, this way facilitating exemplification,
as in Section~\ref{sect:motexa} below.
\section{Some Motivating Examples}\label{sect:motexa}
\newcommand{\enc}{\mathit{ENC}}
\newcommand{\notc}{\mathit{NOT}}
\newcommand{\gen}{\mathit{GEN}}
\newcommand{\expc}{\mathit{EXP}}
\newcommand{\rand}{\mathit{RND}}
\newcommand{\sop}[1]{\overline{#1}}
In this section, we want to show how $\lambda$-calculus can express interesting, although simple, probabilistic programs.
More importantly, we will argue that checking the equivalence of some of the presented programs is not only interesting
from a purely theoretical perspective, but corresponds to a proof of \emph{perfect security} in the sense of Shannon~\cite{S49}.

Let's start from the following very simple programs:
\begin{align*}
  \notc&=\abstr{\varone}{\ifth{\varone}{\underline \ffalse}{\underline \ttrue}}:\arr{\bty}{\bty};\\
  \enc&=\abstr{\varone}{\abstr{\vartwo}{\ifth{\varone}{(\notc\;\vartwo)}{\vartwo}}}:\arr{\bty}{\arr{\bty}{\bty}};\\
  \gen&=\underline\ttrue\oplus\underline\ffalse:\bty.
\end{align*}
The function $\enc$ computes exclusive disjunction as a boolean function, but can also be seen as the encryption function 
of a one-bit version of the so-called One-Time Pad cryptoscheme (OTP in the following). On the other hand, $\gen$ is a term reducing probabilistically
to one of the two possible boolean values, each with probability $\frac{1}{2}$, and is meant to be a way to generate a random key
for the same scheme.

One of the many ways to define perfect security of an encryption scheme consists is setting up an 
\emph{experiment}~\cite{KL07}: the adversary generates two messages, of which one is randomly 
chosen, encrypted, and given back to the adversary who, however, should not be able to guess
whether the first or the second message have been chosen (with success probability strictly greater than $\frac{1}{2}$). 
This can be seen as the problem of proving the following two programs to be context equivalent:
\begin{align*}
  \expc&=\abstr{\varone}{\abstr{\vartwo}{\enc\;(\varone\oplus\vartwo)\;\gen}}:\arr{\bty}{\arr{\bty}{\bty}};\\
  \rand&=\abstr{\varone}{\abstr{\vartwo}{\underline\ttrue\oplus\underline\ffalse}}:\arr{\bty}{\arr{\bty}{\bty}}.
\end{align*}
where $\oplus$ is a primitive for fair, probabilistic choice. Analogously, one could verify that 
any adversary is not able to distinguish an experiment in which the \emph{first} message
is chosen from an experiment in which the \emph{second} message is chosen. This, again, can be seen as the task of
checking whether the following two terms are context equivalent:
\begin{align*}
  \expc_\mathit{FST}&=\abstr{\varone}{\abstr{\vartwo}{\enc\;\varone\;\gen}}:\arr{\bty}{\arr{\bty}{\bty}};\\
  \expc_\mathit{SND}&=\abstr{\varone}{\abstr{\vartwo}{\enc\;\vartwo\;\gen}}:\arr{\bty}{\arr{\bty}{\bty}}.
\end{align*}
But how could we actually \emph{prove} context equivalence? The universal quantification in its definition, as is well known,
turns out to be burdensome in proofs. The task can be made easier by way of various techniques, including context lemmas and
logical relations. Later in this paper, we show how the four terms above can be shown equivalent by way of applicative bisimulation, 
which is proved sound (and complete) with respect to context equivalence in Section \ref{sect:appbis} below.

Before proceeding, we would like to give examples of terms having the same type, but which are \emph{not} context equivalent.
We will do so by again referring to perfect security. The kind of security offered by the OTP is unsatisfactory not only
because keys cannot be shorter than messages, but also because it does not hold in presence of multiple encryptions, or when the 
adversary is \emph{active}, for example by having an access to an encryption oracle. In the aforementioned scenario, security holds
if and only if the following two programs (both of type $\arr{\bty}{\arr{\bty}{\bty\times(\arr{\bty}{\bty})}}$) are context equivalent:
\begin{align*}
  \expc_\mathit{FST}^\mathit{CPA}&=\abstr{\varone}{\abstr{\vartwo}{(\abstr{\varthree}{\pair{\enc\;\varone\;\varthree}{\abstr{\varfour}{\enc\;\varfour\;\varthree}}})\gen}};\\
  \expc_\mathit{SND}^\mathit{CPA}&=\abstr{\varone}{\abstr{\vartwo}{(\abstr{\varthree}{\pair{\enc\;\vartwo\;\varthree}{\abstr{\varfour}{\enc\;\varfour\;\varthree}}})\gen}}.
\end{align*}
It is very easy, however, to realize that if 
$\ctxone=(\abstr{\varone}{(\snd{\varone})(\fst{\varone})})(\emptycon\;\ttrue\;\ffalse)$, then
$\ctxone[\expc_\mathit{FST}^\mathit{CPA}]$ reduces to $\ttrue$, while $\ctxone[\expc_\mathit{SND}^\mathit{CPA}]$ reduces to $\ffalse$, both
with probability $1$. In other words, the OTP is not secure in presence of active adversaries, and for very good reasons: having access
to an oracle for encryption is essentially equivalent to having access to an oracle for \emph{decryption}.
\section{Programs and Their Operational Semantics}
In this section, we will present the syntax and operational semantics of \PCFLP, the language on which we will define applicative
bisimulation. The language \PCFLP\ is identical to Pitts' \PCFL~\cite{Pitts97}, except for the presence of a primitive for binary probabilistic choice.
\subsection{Terms and Types}
The terms of \PCFLP\ are built up from constants (for boolean and integer values, and for the empty list) and variables, 
using the usual constructs from \PCF, and binary choice. In the following, $\vars=\{\varone,\vartwo,\ldots\}$ is a countable set of variables
and $\ops$ is a finite set of binary arithmetic operators including at least the symbols $+$, $\leq$, and $=$.
\begin{definition}
\emph{Terms} are expressions generated by the following grammar: 
\begin{align*}
\termone,\termtwo\bnf& \varone\midd\cnst{n}\midd\cnst{b}\midd\nil\midd\pair{\termone}{\termone}
   \midd\lst{\termone}{\termone}\midd \abstrac{\termone} \midd \fix{\termone}\\ 
   &\midd\termone \oplus\termone\midd \ifth{\termone}{\termone}{\termone} \midd \op{\termone}{\termone} \midd \fst{\termone} \midd \snd{\termone}\\ 
   &\midd \caset{\termone}{\termone}{\termone}\midd  \termone\, \termone,
\end{align*}
where $x,h,t \in\vars$, $n\in\NN$, $b\in\BB=\{\ttrue, \ffalse\}$, $\mathsf{op} \in \ops$. 
\end{definition}
In what follows, we  consider  terms of \PCFLP\  as $\alpha$-equivalence classes of syntax trees. 
The set of free variables of a term $\termone$ is indicated as $\FV{\termone}$. A term $\termone$ is closed if $\FV{\termone} = \emptyset$.
The (capture-avoiding) substitution of $\termtwo$ for the free occurrences of $\varone$ in $\termone$ is denoted 
$\subst{\termone}{\termtwo}{\varone}$. 

The constructions from \PCF\ have their usual meanings. The operator $(\lst{\cdot}{\cdot})$ is the 
constructor for lists, $\nil$ is the empty list, and $\caset{\termthree}{\termone}{\termtwo}$ is a list destructor.
The construct $\termone\oplus\termtwo$ is a binary choice operator, to be
interpreted probabilistically, as in $\LOP$~\cite{DalLagoZorzi}.
\begin{example}
Relevant examples of terms are $\Omega = \left(\fix{x}\right){\underline 0}$, and $\id=\abstrac{x}$: the first one
always diverges, while the second always converges (to itself). In between, one can find terms that converge 
with probability between $0$ and $1$, excluded, e.g., 
$\id\oplus\Omega$, and $\id\oplus\left(\id\oplus\Omega\right)$.
\end{example}
We are only interested in well-formed terms, i.e., terms to which one can assign a type.
\begin{definition}
\emph{Types} are given by the following grammar:
\begin{align*}
\sigma,\tau\bnf&\gamma\midd\arr{\sigma}{\sigma}\midd\cprod{\sigma}{\sigma}\midd\lstty{\sigma};\\
\gamma,\delta\bnf&\bty\midd\ity.
\end{align*}
\end{definition}
The set of all types is $\types$. Please observe that the language of types we consider here coincides
with the one of Pitts' \PCFL~\cite{Pitts97}. An alternative typing discipline 
for probabilistic languages (see, e.g. \cite{RamseyPfeffer}), views probability as a \emph{monad}, this way reflecting the
behaviour of programs in types: if $\sigma$ is a type, $\Box\sigma$ is the type 
of probabilistic distributions over $\sigma$, and the binary choice operator always produces elements of type 
$\Box\sigma$. 
\begin{example}
The following expressions are types: $\ity$,  $\cprod{\ity}{\bty}$, $\arr{\ity}{\left(\cprod{\bty}{\ity}\right)}$.
\end{example}

We assume that all operators from $\ops$ take natural numbers as input, and we associate to each operator $\opc\in\ops$ its 
\emph{result type} $\resty{\gamma}\in\{\bty,\ity\}$ and its semantics $\sop{\opc}:\NN\times\NN\rightarrow X$
where $X$ is either $\BB$ or $\NN$, depending on $\resty{\gamma}$. 
A \emph{typing context} $\Gamma$ is a finite partial function from variables to types. 
$\dom \Gamma$ is the domain of the function $\Gamma$. If $x\not\in\dom{\Gamma}$, 
$\left(x:\sigma,\Gamma\right)$ represents the function which extends $\Gamma$ 
to $\dom{\Gamma}\cup\{\varone\}$, by associating $\sigma$ to $\varone$.
\begin{definition}
A \emph{typing judgement} is an assertion of the form
$\Gamma \proves \termone: \sigma$, where $\Gamma$ is a context, $\termone$ is a term, and $\sigma$ is a type. 
A judgement is valid if it can be derived by the rules of the formal system given in Figure~\ref{fig1}.
\end{definition}
\begin{center}
\begin{figure}[htpb]
\fbox{
\begin{minipage}{.97\textwidth}
{\center
\def\proofSkipAmount{\vskip 10ex plus.10ex minus.5ex}
\small
$$
  \AxiomC{  $(x,\sigma) \in \Gamma$}
  \UnaryInfC{$\Gamma \proves x: \sigma$}
  \DisplayProof 
  \qquad
  \AxiomC{ $b \in \mathbb{B}$}
  \UnaryInfC{$\Gamma \proves \underline{b}: \bty$}
  \DisplayProof 
  \qquad
  \AxiomC{ $n \in \mathbb{N}$}
  \UnaryInfC{$\Gamma \proves \underline{n}: \ity$}
  \DisplayProof 
  \qquad
  \AxiomC{$\Gamma \proves \termone: \sigma$}
  \AxiomC{$\Gamma \proves \termtwo: \sigma$}
  \AxiomC{$\Gamma \proves \termthree: \bty$}
  \TrinaryInfC{$\Gamma \proves \ifth{\termthree}{\termone}{\termtwo}$ }
  \DisplayProof 
$$
$$ 
  \AxiomC{$\Gamma \proves \termone: \ity$}
  \AxiomC{$\Gamma \proves \termtwo: \ity$}
  \BinaryInfC{$\Gamma \proves  {\op {\termone} {\termtwo}}: \gamma_{\mathsf{op}}$ }
  \DisplayProof 
  \qquad 
  \AxiomC{$\Gamma, x: \sigma \proves \termone: \tau$}
  \UnaryInfC{$\Gamma \proves \abstrac \termone: \arr{\sigma}{\tau}$ }
  \DisplayProof 
  \qquad
  \AxiomC{$\Gamma \proves \termone: \sigma$}
  \AxiomC{$\Gamma \proves \termtwo: \sigma$}
  \BinaryInfC{$\Gamma \proves \termone\oplus\termtwo: \sigma$ }
  \DisplayProof 
$$
$$
  \AxiomC{$\Gamma \proves\termone: \arr{\sigma}{\tau}$}
  \AxiomC{$\Gamma \proves\termtwo: \sigma$}
  \BinaryInfC{$\Gamma \proves\termone\termtwo: \tau$ }
  \DisplayProof 
  \qquad
  \AxiomC{$\Gamma, x:\arr{\sigma}{\tau} \proves\termone:\arr{\sigma}{\tau}$}
  \UnaryInfC{$\Gamma \proves \fix\termone:\arr{\sigma}{\tau}$ }
  \DisplayProof 
  \qquad
  \AxiomC{$\Gamma \proves \termone: \sigma$}
  \AxiomC{$\Gamma \proves \termtwo: \tau$}
  \BinaryInfC{$\Gamma \proves \pair {\termone} {\termtwo}: \sigma \times \tau$ }
  \DisplayProof 
$$
$$
  \AxiomC{$\Gamma \proves \termone: \sigma \times \tau$}
  \UnaryInfC{$\Gamma \proves \fst \termone: \sigma$ }
  \DisplayProof 
  \qquad
  \AxiomC{$\Gamma \proves \termone: \sigma \times \tau$}
  \UnaryInfC{$\Gamma \proves \snd \termone: \tau$ }
  \DisplayProof 
  \qquad
  \AxiomC{}
  \UnaryInfC{$\Gamma \proves nil: [\sigma]$ }
  \DisplayProof \quad
  \AxiomC{$\Gamma \proves T: [\sigma]$}
  \AxiomC{$\Gamma \proves H: \sigma$}
  \BinaryInfC{$\Gamma \proves H:: T: [\sigma]$ }
  \DisplayProof
$$
$$
  \AxiomC{$\Gamma \proves L: [\sigma]$}
  \AxiomC{$\Gamma \proves \termone_1: \tau$}
  \AxiomC{$\Gamma, h:\sigma, t:[\sigma]  \proves \termone_2: \tau$}
  \TrinaryInfC{$\Gamma \proves \caset L {\termone_1} {\termone_2}: \tau$ }
  \DisplayProof
$$}
\end{minipage}}
\caption{Type Assignment in \PCFLP}
\label{fig1}
\end{figure}
\end{center}
Please notice that any term of which we want to form the fixpoint needs to be a function.
\begin{definition}
If $\sigma$ is a type and $\Gamma$ is a typing context, then 
$\termty{\sigma} = \{t\,|\,\emptyset \proves t:\sigma  \}$,  
$\termty{} = \{t\,|\,\exists \sigma, t \in \termty \sigma\}$, 
$\termtyg{\Gamma}{\sigma} = \{t\, |\, \Gamma \proves t:\sigma  \}$.
\end{definition}
In other words, $\termty{\sigma}$ is the set of \emph{closed terms} (also called \emph{programs}) 
of type $\sigma$, while $\termty{}$ is the set of closed terms which have a valid typing 
derivation, and $\termtyg{\Gamma}{\sigma}$ is the set of terms which have type 
$\sigma$ under the context $\Gamma$. We can observe that $\termty\sigma=\termtyg\emptyset \sigma$.  
\begin{example}
The following type assignments are valid:
\begin{varitemize}
\item $\forall \Gamma \text{ a context , and }  \sigma \text{ a type: }$ $\Gamma \proves I: \sigma \rightarrow \sigma$;
\item For every function type $\tau$, and all typing context $\Gamma$,  $\Gamma \proves \fix x: \tau$;
\item The previous point allow us to see that for all type $\sigma$, ${\emptyset \proves \fix x:\ity \rightarrow \sigma}$.So $\forall \Gamma\text{ a context , and }  \sigma\text{ a type: }$ $\Gamma \proves \Omega: \sigma$;
\item $\forall \Gamma\text{ a context , and }  \sigma\text{ a type: }$ $\Gamma \proves I \oplus \Omega: \sigma \rightarrow \sigma$;
\item $\emptyset \proves \fix \left((\lambda z. \underline{0}) \oplus \lambda z. ((x\, \underline{0})+\underline{1})\right): \ity \rightarrow \ity$. 
\end{varitemize}
\end{example}
\subsection{Operational Semantics}
Because of the probabilistic nature of choice in \PCFLP, a program 
doesn't evaluate to a value, but to a probability distribution of values. 
Therefore, we need the following notions 
to define an evaluation relation.
\begin{definition}
\emph{Values} are terms of the following form:
$$
\valone\bnf\cnst{n}\midd\cnst{b}\midd\nil\midd\abstrac{\termone}\midd\fix{\termone}\midd\lst{\termone}{\termone}\midd\pair{\termone}{\termone}.
$$
We will call \values\ the set of values, and we note 
$\val{\sigma}=\values\cap\termty{\sigma}$. A \emph{value distribution} 
is a function $\distrone:$ \values$\rightarrow [0,1]$, such that $\sum_{\valone\in\values}\distrone(\valone) \leq 1$.
Given a value distribution $\distrone$, we will note $\supp{\distrone}$ the set of those values $\valone$ such that 
$\distrone(\valone) > 0$. A value distribution $\distrone$ is said \emph{finite} whenever $\supp{\distrone}$ 
has finite cardinality. If $\valone$ is a value, we note $\distrv{\valone}$ the value distribution 
$\distrone$ such that $\distrone(\valtwo)=1$ if $\valtwo=\valone$ and $\distrone(\valone)=0$ 
otherwise. Value distributions can be ordered pointwise.
\end{definition}
We first give an \emph{approximation} semantics, which attributes \emph{finite} probability distributions to terms, 
and only later define the actual semantics, which will be the least upper bound of all distributions obtained through
the approximation semantics. Big-step semantics is given by way of a binary relation $\Downarrow$ 
between closed terms and value distributions, which is defined by 
the set of rules from Figure \ref{fig:cbvevalrules}.
\begin{figure}[htpb]
\fbox{
\begin{minipage}{.99\textwidth}
{\center
$$
\AxiomC{}	
\RightLabel{\footnotesize $b_e$}
\UnaryInfC{\small $ \evav{\termone}{\emptyset} $ }
\DisplayProof 
\qquad
\AxiomC{}	
\RightLabel{\footnotesize $b_v$}
\UnaryInfC{\small $ \evav{\valone}{\distrv{\valone}} $ }
\DisplayProof
\qquad
\AxiomC{\small $\evalv {\termone}{\distrone}$}	
\AxiomC{\small$\evalv {\termtwo}{\distrtwo}$}
\RightLabel{$\footnotesize b_{op}$}
\BinaryInfC{\small$ \evav {\op{\termone}{\termtwo}} {\sum\limits_{\underline{n} \in S(\distrone),\underline{m} \in S(\distrtwo)}
    \distrone(\underline{n})\distrtwo(\underline{m})\distrv{\underline{\opn{n}{m}  }}} $ }
\DisplayProof 
$$
$$
\AxiomC{\small$\evalv{\termone}{K}$}
\AxiomC{\small$\evalv{\termtwo}{F}$}	
\AxiomC{\small$\{ {\evalv {\subst P v x} {E}} _{P,v}\}_{\abstrac P \in S(\distr K), \, v \in S(\distr F)}$}
\noLine
\TrinaryInfC{\small$\{{\evalv {\subst Q {\fix Q}{x}v } {G}} _{Q,v} \}_{\fix Q \in S(\distr{K}), \, v \in S(\distr F)}$}
\RightLabel{$\footnotesize b_a$}
\UnaryInfC{\small$ 
  \evav {\termone\termtwo}  {\sum_{v \in S(\distr F)} \distr{F} (v) \left( \sum_{\abstrac P \in S(\distr{K})} \distr{K}(\abstrac P).\distr{E}_{P,v}  + 
  \sum_{\fix Q \in S(\distr{K})} \distr{K}(\fix Q) .\distr{G}_{Q,v}\right)}$}
\DisplayProof 
$$
$$
\AxiomC{\small $\evalv {L} D$}
\AxiomC{\small$\evalv {M_1} {E_1}$}	
\AxiomC{\small$\evalv {M_2} {E_2}$}
  \RightLabel{$\footnotesize b_{if}$}
  \TrinaryInfC{\small$ \evav {\ifth L {M_1} {M_2} } \distr{D} (\underline{\ttrue})\distr{E_1} + \distr{D}(\underline{\ffalse}) \distr{E_2} $ }
\DisplayProof 
$$
{\scriptsize
$$
\AxiomC{\small$\evalv {L} D$}
\AxiomC{\small$\evalv {M_1} E$}
\AxiomC{\small$(\evalv H G{_{(H,T)}}, \, \evalv T K
{_{(H,T)}})_{H::T \in S(\distrone) }$}	
\AxiomC{\small$\{\evalv {M_2\{V/h,W/t\}} F {_{V,W}}\}_{\substack{H::T \in S(\distrone) , \,V \in S(\distr G_{(H,T)}),\\  W \in S(\distr {K}_{(H,T)})}}$}
  \RightLabel{\small$\footnotesize b_{case}$}
  \QuaternaryInfC{\small$ \evav {\caset L {M_1} {M_2}} {\distr{D}(nil)\distr{E} + \sum\limits_{H::T \in S(\distrone)}\sum\limits_{\substack{V \in S(\distr G_{(H,T)})\\ W \in S(\distr {K}_{(H,T)}) }}\distr{D}(H::T) \distr {G}_{(H,T)}(V) \distr {K}_{(H,T)}(W) \distr{F}_{V,W}} $ }
\DisplayProof
$$}
$$
\AxiomC{\small$\evav {M} {\distr{D}}$}
\AxiomC{\small$\{\evav {P} {\distr{E}_P}\}_{\pair P N \in S(\distr{D})}$}
  \RightLabel{$\footnotesize b_{fst}$}
  \BinaryInfC{\small$ \evav {\fst M} {\sum_{\pair P N \in S(\distr{D})} \distr{D}(\pair P N). \distr{E}_P } $}
  \DisplayProof \quad
\AxiomC{\small$\evalv {M} D$}
\AxiomC{\small$\{\evav {N} {\distr{E}_N}\}_{\pair P N \in S(\distr{D})}$}
  \RightLabel{$\footnotesize b_{snd}$}
  \BinaryInfC{\small$ \evav {\snd M} {\sum_{\pair P N \in S(\distr{D})} \distr{D}(\pair P N). \distr{E}_N } $}
  \DisplayProof
$$
$$
\AxiomC{\small$\evalv {M_1} D{_1}$}	
\AxiomC{\small$\evalv {M_2} D{_2}$}
  \RightLabel{$\footnotesize b_s$}
  \BinaryInfC{\small$ \evav {M_1 \oplus M_2} {\frac{1}{2}\distr{D}_1} + \frac{1}{2}\distr{D}_2 $ }
\DisplayProof
$$
}
\end{minipage}}
\caption{Evaluation --- Rule Selection}
\label{fig:cbvevalrules}
\end{figure}
This evaluation relation, by the way, is the natural extension to \PCFLP\ of the evaluation relation given 
in \cite{DalLagoZorzi} for the untyped probabilistic $\lambda$-calculus. 
Please observe how function arguments are evaluated before being passed to functions. Moreover, 
$\lst{\termone}{\termtwo}$ is a value even if $\termone$ or $\termtwo$ are not, which means that 
lists are \emph{lazy} and potentially infinite.
\begin{proposition}
Call-by-value evaluation preserves typing, that is:
if $\evalv{\termone}{\distrone}$, and $\termone\in\termty{\sigma}$, then for every 
$\valone\in\supp{\distrone}$, $\valone\in\val{\sigma}$. 
\end{proposition}
\begin{lemma}\label{lem:directed}
  For every term $\termone$, if $\evalv{\termone}{\distrone}$, 
  and $\evalv{\termone}{\distrtwo}$, then there exists a distribution 
  $\distrthree$ such that $\evalv{\termone}{\distrthree}$ with 
  $\distrone\leq\distrthree$, and $\distrtwo\leq\distrthree$.  
\end{lemma}
\begin{proof}
The proof is by induction on the structure of derivations for $\evalv \termone {\distrone}$.
\end{proof}

\begin{definition}
  For any closed term $\termone$, we define the \emph{big-steps semantics} 
  $\diS{\termone}$ of $\termone$ as $\sup_{\evalv{\termone}{\distrone}}\distrone$.
\end{definition}
Since distributions form an $\omega$-complete partial order, and for every $\termone$
the set of those distributions $\distrone$ such that $\evalv{\termone}\distrone$ is a countable  
directed set (by Lemma \ref{lem:directed}), this definition is well-posed,
and associates a unique value distribution to every term.
In \cite{DalLagoZorzi} various ways to define coinductively call-by-value approximation 
semantics on probabilistic untyped $\lambda$-calculus were introduced, and it was proved that the semantics obtained by taking the 
greatest lower bound of this coinductive approximation semantics was equivalent to the inductively characterized semantics. 
Although it is possible to extend similarly those definitions for \PCFLP\, we do not do it, and only limit our
attention to inductively defined probabilistic semantics.

The distribution $\diS{\termone}$ can be obtained equivalently by taking the least upper bound
of all finite distributions $\distrone$ for which $\evalvs{\termone}{\distrone}$, where $\Rightarrow$ is a
binary relation capturing \emph{small-step} evaluation of terms. 
This proceeds as follows. The first step consists in defining the notion of an
\emph{evaluation context}, which in the case of \PCFLP\ is the following one
\begin{align*}
  \ectxone::=&\emptycon\midd \ectxone \termone \midd \valone \ectxone 
     \midd \op{\ectxone}{\termone} \midd \op{\valone}{\ectxone} \midd 
     \midd \fst \ectxone \midd \snd \ectxone \midd \ifth{\ectxone}{\termone}{\termone}\\
     &\midd\caset{\ectxone}{\termone}{\termone}
      \midd\caset{\ectxone::\termone}{\termone}{\termone}\\
     &\midd\caset{\valone::\ectxone}{\termone}{\termone}
\end{align*}
The next step consists in giving a relation modelling one-step reduction. In our step
this takes the form of a relation $\SSred$ between closed terms and \emph{sequences}
of closed terms, which is defined as the smallest such relation including satisfying
the following rules:
\begin{align*}
  (\abstr{\varone}{\termone})\valone&\SSred\subst{\termtwo}{\valone}{\varone};\\
  (\fix{\termone})\valone&\SSred(\subst{\termone}{(\fix{\termone})}{\varone})\valone;\\
  \op{\underline{n}}{\underline{m}}&\SSred\underline{\mathsf{op}(m,n)};\\
  \fst{\pair{\valone}{\valtwo}}&\SSred\valone;\\
  \snd{\pair{\valone}{\valtwo}}&\SSred\valtwo;\\
  \ifth{\ttrue}{\termone}{\termtwo}&\SSred\termone;\\
  \ifth{\ffalse}{\termone}{\termtwo}&\SSred\termtwo;\\
  \caset{\nil}{\termone}{\termtwo}&\SSred\termone;\\
  \caset{\valone::\valtwo}{\termone}{\termtwo}&\SSred\subst{\termtwo}{\valone,\valtwo}{h,t}.
\end{align*}
and closed under all evaluation contexts, i.e., if $\termone\SSred\termtwo_1,\ldots,\termtwo_n$,
then we also have that $\ectxone[\termone]\SSred\ectxone[\termtwo_1],\ldots,\ectxone[\termtwo_n]$.
Proper probabilistic computation enters the playground as soon as we define the relation $\Rightarrow$
between closed terms and value distributions, which is defined by inductively interpreting the
following three rules:
$$
\AxiomC{}	
\UnaryInfC{\small $\evalvs{\termone}{\emptyset} $ }
\DisplayProof 
\qquad
\AxiomC{}	
\UnaryInfC{\small $\evalvs{\valone}{\{\valone^1\}} $ }
\DisplayProof
\qquad
\AxiomC{\small $\termone\SSred\termtwo_1,\ldots,\termtwo_n$}
\AxiomC{\small $\evalvs{\termtwo_i}{\distrone_i}$}
\BinaryInfC{\small $\evalvs{\termone}{\sum_{1\leq i\leq n}}\frac{1}{n}\distrone_i$}
\DisplayProof
$$
\begin{theorem}[Big-step is Equivalent to Small-step~\cite{DalLagoZorzi}]
$\diS{\termone}=\sup_{\evalvs{\termone}{\distrone}}\distrone$.
\end{theorem}
\begin{example}
Approximation semantics does not allow to derive any assertion about $\Omega$, and indeed
$\diS{\Omega}=\emptyset$. Similarly, $\diS{\id}=\{\id^1\}$.
Recursion allows to define much more interesting programs, e.g.
$\termone=\left(\fix{(\abstr{\vartwo}{\vartwo})\oplus\abstr{\vartwo}{x(\vartwo+\underline{1})}}\right)\underline {0}$. 
Indeed, $\diS{\termone}(\underline{n})=\frac{1}{2^{n+1}}$ for every $n\in\NN$, even if
$\termone\not\Downarrow\diS{\termone}$.
As another example, $\diS{(\abstrac{I} \oplus \abstrac{\Omega}) \underline{0}} = \frac{1}{2}\distrv{I}$.
Finally, $\diS{(\fix{I\oplus x}){\underline 0}} = \distrv {\underline 0}$, but please observe that we don't have 
$\evalv {\left(\fix{\left(I \oplus x\right)}\right) {\underline 0}}{\{\underline{0}^1\}}$.
\end{example}
\subsection{Relations}
The notion of \emph{typed relation} corresponds to a family of
relations $(\relone_\sigma^\Gamma)_{\sigma,\Gamma}$, each of them
a binary relation on $\termtyg{\sigma}{\Gamma}$. We extend the usual notion of symmetry, reflexivity and transitivity to 
typed relations in the following way:
\begin{definition}
A \emph{typed relation} is a family $\relone=(\relone_\sigma^\Gamma)_{\sigma,\Gamma}$, where each 
$\mathcal{R}_\sigma^\Gamma$ is a binary relation on $\termtyg \sigma \Gamma$. Sometime,
$\rela{\termone}{\relone_\sigma^\Gamma}{\termtwo}$ will be noted as $\Gamma \proves\rela{\termone}{\relone_\sigma}{\termtwo}$ 
(or as $\Gamma \proves\relonea{\termone}{\termtwo}:\sigma$). A typed relation $\relone$ is said to be:
\begin{varitemize}
\item 
  \emph{reflexive} if 
  $\forall \termone \in \termtyg{\sigma}{\Gamma}$ it holds that $\Gamma \proves\relonea{\termone}{\termone}:\sigma$;
 \item 
   \emph{symmetric} if  
   $\forall \sigma,\Gamma.\,\forall \termone, N \in \termtyg \sigma \Gamma.\,\Gamma \proves 
   \termone {\relone} N:\sigma \Rightarrow \Gamma \proves \termtwo {\relone} \termone: \sigma$;
 \item 
   \emph{transitive} if 
   $\forall \sigma,\Gamma .\, \forall \termone, N,L \in {\termtyg \sigma \Gamma} . \,\left(\Gamma \proves 
   \termone {\relone} N: \sigma \wedge \Gamma \proves N {\relone} L: \sigma \right) \Rightarrow 
   \Gamma \proves \termone {\relone} L: \sigma$.
\end{varitemize}
\end{definition}
\begin{definition}
Let $\relone$ be a typed relation. We define the \emph{compatibility} of 
$\relone$  in the expected way. For instance, if $\relone$ is compatible, the following properties should hold:  
\begin{varitemize}
\item 
  $\Gamma \proves \underline{k}\,{\relone}\, \underline{k}: \ity$ for every $k \in \mathbb{N}$;
\item 
  $x:\tau,\Gamma \proves x\, {\relone} \,x\,:\,\tau$ for every $\varone$ and for every $\tau$;
\item 
  $\Gamma \proves \termone {\relone}\termtwo:\sigma$ and 
  $\Gamma \proves \termthree {\relone}\termfour:\sigma$ implies
  $\Gamma \proves  ({\termone} \oplus {\termthree})\,{\relone}\,({\termtwo} \oplus {\termfour})  \,:\, \sigma $.
\end{varitemize}
\end{definition}
Please observe that a compatible typed relation $\relone$ is always reflexive, since $\mathcal{R}$ is reflexive 
for terms of ground form, and that $\relone$ is stable by the constructors of the language:
\begin{proposition}
Let ${\relone}$ be a typed relation. If ${\relone}$ is compatible, then ${\relone}$ is reflexive. 
\end{proposition}
Any typed relation capturing a notion of equivalence should be a congruence, this way being applicable
at any point in the program, possibly many times:
\begin{definition}
Let ${\relone}$ be a typed relation.
Then ${\relone}$ is said to be a \emph{precongruence relation} if ${\relone}$ is transitive and compatible, and
${\relone}$ is said to be a \emph{congruence relation} if ${\relone}$ is symmetric, transitive  and compatible. 
\end{definition}
We write $\relct$ for the set of type-indexed families $\relone = (\mathcal{R}_\sigma)_\sigma$ of binary relations
$\mathcal{R}_\sigma$ between the terms in $\termty{\sigma}$.
\subsection{Context Equivalence}
The general idea of context equivalence is the following: two terms $\termone$ and $\termtwo$ are equivalent if any 
occurrence  of $\termone$ in any program $\termthree$ can be replaced with $\termtwo$  without changing 
the observable behaviour of $\termthree$. The notion of a context allows us to formalize this idea.
\begin{definition}
A \emph{context} is a syntax tree with a unique hole:
\begin{align*}
  \ctxone::=&\emptycon \midd \abstrac \ctxone \midd \ctxone \termone \midd \termone \ctxone 
     \midd \ctxone \oplus \termone \midd\termone \oplus \ctxone\\
  \midd &\fix \ctxone \midd {\op \termone \ctxone} \midd {\op \ctxone \termone} \midd 
     \pair \ctxone \termone \midd \pair \termone \ctxone \midd \fst \ctxone \midd \snd \ctxone \\
  \midd &\ifth \termone \termone \ctxone \midd \ifth \termone \ctxone \termone \midd 
     \ifth \ctxone \termone \termone \midd \termone::\ctxone \midd \ctxone::\termone \\
  \midd &\caset \ctxone \termone \termone \midd \caset \termone \termone \ctxone \midd 
     \caset \termone \ctxone \termone   
\end{align*}
Given a context $\ctxone$ and a term $\termone$, $\act{\ctxone}{\termone}$ is the term obtained by
substituting the unique hole in $\ctxone$ with $\termone$.
\end{definition}
When defining context equivalence, we work with \emph{closing} contexts, namely those contexts 
$\ctxone$ such that $\ctxone[\termone]$, and $\ctxone[\termtwo]$ are closed terms (where $\termone$
and $\termtwo$ are the possibly open terms being compared).
We are now going to define a notion of typing for contexts. Judgments have the shape
$\Gamma \proves \tycon{\ctxone}{\Delta}{\sigma}{\tau}$, which informally means that if $\termone$ is a term of type 
$\sigma$ under the typing context  $\Delta$, then the hole of $\ctxone$ can be filled by $\termone$, obtaining 
a term of type $\tau$ in the context $\Gamma$. So in order to define well-typed substitutions 
for a term in a context, we extend, in a natural way, the notion of typing to contexts:
\begin{definition}
  A \emph{typing judgement for contexts} is an assertion of the form:
  $\Gamma \proves\tycon{\ctxone}{\Delta}{\sigma}{\tau}$, where $\Gamma$, $\Delta$ are typing contexts, 
  $\ctxone$ is a context, and $A$, $B$ are types.  A judgement is valid if it can be derived by the rules of the 
  formal system given in Figure \ref{fig:typcon}.
\end{definition}
\begin{figure}[htpb]
\fbox{
\begin{minipage}{.97\textwidth}
{\small
\center
\begin{prooftree}
\AxiomC{$\dom{\Gamma}\cap\dom{\Delta}=\emptyset$}
\UnaryInfC{$\Gamma \proves \emptycon(\Delta;A): A$}
\end{prooftree}
\begin{prooftree}
\AxiomC{$\Gamma,x:\sigma \proves  C \,(\Delta;B): \tau$}
\AxiomC{$x \not \in dom(\Gamma), x \not \in dom(\Delta)$}
\BinaryInfC{$\Gamma \proves {\abstrac C}\,(x:\sigma\, ,\, \Delta\,;\,B): \sigma \rightarrow \tau$}
\end{prooftree}
\begin{prooftree}
\AxiomC{$\Gamma,x:\sigma \proves  C \,(\Delta;B): \tau$}
\RightLabel{$x \not \in dom (\Gamma), x \not \in dom(\Delta) \atop \tau\text{ a function type }$}
\RightLabel{{\small $fix$}}
\UnaryInfC{$\Gamma \proves {\fix C}\,(x:\tau\, ,\, \Delta\,;\,B): \tau$}
\end{prooftree}
\begin{align*}
\AxiomC{$\Gamma \proves  C \,(\Delta;B):\sigma \rightarrow \tau$}
\AxiomC{$\Gamma \proves M: \sigma $}
\BinaryInfC{$\Gamma \proves  C M \, (\Delta\,;\,B): \tau$}
\DisplayProof & &
\AxiomC{$\Gamma \proves  C \,(\Delta;B):\sigma $}
\AxiomC{$\Gamma \proves M: \sigma \rightarrow \tau$}
\BinaryInfC{$\Gamma \proves  M C \, (\Delta\,;\,B): \tau$}
\DisplayProof
\end{align*}
\begin{align*}
\AxiomC{$\Gamma \proves  C \,(\Delta;B):\sigma \times \tau$}
\UnaryInfC{$\Gamma \proves  \fst C  \, (\Delta\,;\,B): \sigma$}
\DisplayProof & &
\AxiomC{$\Gamma \proves  C \,(\Delta;B):\sigma \times \tau$}
\UnaryInfC{$\Gamma \proves  \snd C  \, (\Delta\,;\,B): \tau$}
\DisplayProof
\end{align*}
\begin{align*}
\AxiomC{$\Gamma \proves  C \,(\Delta;B):\sigma$}
\AxiomC{$\Gamma \proves M: \tau $}
\BinaryInfC{$\Gamma \proves   \pair C M  \, (\Delta\,;\,B): \sigma \times \tau $}
\DisplayProof & &
\AxiomC{$\Gamma \proves  C \,(\Delta;B):\sigma$}
\AxiomC{$\Gamma \proves M: \tau $}
\BinaryInfC{$\Gamma \proves   \pair M C \, (\Delta\,;\,B): \tau \times \sigma$}
\DisplayProof
\end{align*}
\begin{align*}
\AxiomC{$\Gamma \proves  C \,(\Delta;B):\sigma$}
\AxiomC{$\Gamma \proves M: [\sigma] $}
\BinaryInfC{$\Gamma \proves   C:: M  \, (\Delta\,;\,B): [\sigma]$}
\DisplayProof
& &
\AxiomC{$\Gamma \proves  C \,(\Delta;B):[\sigma]$}
\AxiomC{$\Gamma \proves M: \sigma $}
\BinaryInfC{$\Gamma \proves   M:: C  \, (\Delta\,;\,B): [\sigma]$}
\DisplayProof
\end{align*}
\begin{align*}
\AxiomC{$\Gamma \proves  C \,(\Delta;B):\sigma$}
\AxiomC{$\Gamma \proves M: \sigma $}
\BinaryInfC{$\Gamma \proves   M \oplus C  \, (\Delta\,;\,B): \sigma$}
\DisplayProof & &
\AxiomC{$\Gamma \proves  C \,(\Delta;B):\sigma$}
\AxiomC{$\Gamma \proves M: \sigma $}
\BinaryInfC{$\Gamma \proves   C \oplus M  \, (\Delta\,;\,B): \sigma$}
\DisplayProof
\end{align*}
\begin{align*}
\AxiomC{$\Gamma \proves  C \,(\Delta;B):\ity$}
\AxiomC{$\Gamma \proves M: \ity $}
\BinaryInfC{$\Gamma \proves {\op M C}  \, (\Delta\,;\,B): \gamma_\mathsf{op}$}
\DisplayProof & &
\AxiomC{$\Gamma \proves  C \,(\Delta;B):\ity$}
\AxiomC{$\Gamma \proves M: \ity $}
\BinaryInfC{$\Gamma \proves {\op{C}{\termone}}\, (\Delta\,;\,B): \gamma_\mathsf{op} $}
\DisplayProof
\end{align*}
\begin{prooftree}
\AxiomC{$\Gamma \proves  C \,(\Delta;B):\bty$}
\AxiomC{$\Gamma \proves M_1: \sigma $}
\AxiomC{$\Gamma \proves M_2: \sigma $}
\TrinaryInfC{$\Gamma \proves  \ifth  C {M_1} {M_2}  \, (\Delta\,;\,B): \sigma $}
\end{prooftree}
\begin{prooftree}
\AxiomC{$\Gamma \proves  C \,(\Delta;B):\sigma$}
\AxiomC{$\Gamma \proves L: \bty $}
\AxiomC{$\Gamma \proves M: \sigma $}
\TrinaryInfC{$\Gamma \proves  \ifth  L C M   \, (\Delta\,;\,B): \sigma $}
\end{prooftree}
\begin{prooftree}
\AxiomC{$\Gamma \proves  C \,(\Delta;B):\sigma$}
\AxiomC{$\Gamma \proves L: \bty $}
\AxiomC{$\Gamma \proves M: \sigma $}
\TrinaryInfC{$\Gamma \proves  \ifth  L M C  \, (\Delta\,;\,B): \sigma $}
\end{prooftree}
\begin{prooftree}
\AxiomC{$\Gamma \proves  C \,(\Delta;B):[\sigma]$}
\AxiomC{$\Gamma \proves M_1: \tau$}
\AxiomC{$\Gamma,h:\sigma, \, t: [\sigma] \proves M_2: \tau $}
\RightLabel{$h,t \not \in dom(\Gamma) $}
\TrinaryInfC{$\Gamma \proves  \caset C {M_1} {M_2} (\Delta;B): \tau $}
\end{prooftree}
\begin{prooftree}
\AxiomC{$\Gamma \proves  C \,(\Delta;B):\tau$}
\AxiomC{$\Gamma \proves M_1: [\sigma]$}
\AxiomC{$\Gamma,h:\sigma, \, t: [\sigma] \proves M_2: \tau $}
\RightLabel{$h,t \not \in dom(\Gamma) $}
\TrinaryInfC{$\Gamma \proves  \caset {M_1} C {M_2} (\Delta;B): \tau$}
\end{prooftree}
\begin{prooftree}
\AxiomC{$\Gamma, h:\sigma, \, t: [\sigma] \proves  C \,(\Delta;B):\tau$}
\AxiomC{$\Gamma \proves M_1: [\sigma]$}
\AxiomC{$\Gamma \proves M_2: \tau $}
\RightLabel{$h,t \not \in dom(\Gamma) \cup dom( \Delta) $}
\TrinaryInfC{$\Gamma \proves  \caset {M_1} {M_2} C (\Delta,h:\sigma,t:[\sigma]\,;\,B): \tau$}
\end{prooftree}}
\end{minipage}}
\caption{Context Type Assignment}
\label{fig:typcon}
\end{figure}
The operation $\termone \mapsto C[\termone]$ of substituting a \PCFLP\ term for a parameter 
in a context to obtain a new \PCFLP\ term respect typing in the following sense:
\begin{proposition}
Let be $\Gamma$, $\Delta$, such that $\dom \Gamma \cap \dom \Delta = \emptyset$.
Let be $\termone$ such that $\Gamma,\Delta\proves \termone: \sigma$, and $C$ a context such that: 
$\Gamma \proves C (\Delta; \sigma): \tau$.  Then $\Gamma \proves C[\termone]: \tau$.
\end{proposition}
\begin{proof}
The proof is by induction on the structure of the derivation of $\Gamma \proves C (\Delta, \sigma): \tau$.
\end{proof}
\begin{example}
Example of derivable judgments of the just described form are
$\emptyset \proves \abstrac {\emptycon} \, (x:\sigma;\tau): (\sigma \rightarrow \tau)$
and $\emptyset \proves \left(\left(\abstrac {\underline{\ttrue}} \right)\emptycon\right)\, (\emptyset;\sigma): \bty$.
\end{example}
Here, following \cite{DanosHarmer,ADLS13,EPT13}, we consider that the observable behaviour of a program $\termone$ is its \emph{probability of convergence} 
$\sum\diS{\termone}=\sum_{\valone} \diS{\termone}(\valone)$. We now have all the ingredients necessary to define what context
equivalence is:
\begin{definition}
  The \emph{ contextual preorder } is the typed relation $\leq$ given by: 
  for every $M, N \in \termtyg \tau \Gamma$,  $\Gamma \proves M \leq N: \tau $ if for every context 
  $C$ such that $\emptyset \proves C \,(\Gamma;\tau): \sigma$, it holds
  that $\sum\diS{\act{C}{\termone}}\leq\sum\diS{\act{C}{\termtwo}}$.
  \emph{Context equivalence} is the typed relation $\equiv$ given by stipulating that  
  $\Gamma \proves M \equiv  N: \sigma$ iff $\Gamma \proves M \leq N: \sigma$ 
  and $\Gamma \proves N \leq M: \sigma$.
\end{definition}
Another way to define context equivalence would be to restrain ourselves to contexts of $\bty$ and $\ity$ type in 
the definition of context equivalence: this is the so-called \emph{ground} context equivalence. In a call-by-value
setting, however, this gives exactly the same relation, since any non-ground context can be turned into a ground context inducing
the same probability of convergence. A similar argument holds for a notion of equivalence in which one observes the obtained
(ground) \emph{distribution} rather than merely its sum. The following can be proved in a standard
way: 
\begin{proposition}
$\leq$ is a typed relation, which is reflexive, transitive and compatible.
\end{proposition}

Because of the quantification over all contexts, it is usually difficult to show that $\termone$ and $\termtwo$ are two context 
equivalent terms. In the next sections, we will introduce another notion of equivalence, and we show that it is 
included in context equivalence. 
\section{Applicative Bisimulation}\label{sect:appbis}
In this section, we introduce the notions of similarity and bisimilarity for \PCFLP.
We proceed by instantiating probabilistic bisimulation as developed by Larsen and Skou for a generic labelled
Markov chain in \cite{LS91}. A similar use was done for a call-by-name untyped probabilistic $\lambda$-calculus 
$\Lambda_\oplus$ in \cite{ADLS13}. 
\subsection{Larsen and Skou's Probabilistic Bisimulation}
Preliminary to the notion of (bi)simulation, is the notion of a \emph{labelled Markov chain} (LMC in the following), which is a triple
$\mcone=(\mathcal{S},\mathcal{L},\mathcal{P})$, where $\mathcal{S}$ is a countable set of \emph{states}, $\mathcal{L}$ is a set of 
\emph{labels}, and $\mathcal{P}$ is a \emph{transition probability matrix}, i.e., a function $\mathcal{P}:\mathcal{S}\times\mathcal{L}\times\mathcal{S}\rightarrow\RR$
such that for every state $\stateone\in\mathcal{S}$ and for every label $l\in\mathcal{L}$, $\sum_{\statetwo\in\mathcal{S}}\mathcal{P}(\stateone,l,\statetwo)\leq 1$.
Following~\cite{DEP02}, we allow the sum above to be smaller than $1$, modelling divergence this way.
The following is due to Larsen and Skou~\cite{LS91}:
\begin{definition}\label{def:probasim}
  Given $(\mathcal S,\mathcal L, \mathcal P)$ a labelled Markov Chain, a \emph{ probabilistic simulation} is a pre-order relation $\vrelone$ 
  on $\mathcal{S}$ such that $(s,t)\in \vrelone$ implies that for every $X \subseteq\mathcal{S}$ and for every $l\in\mathcal{L}$, 
  $\mathcal{P}(s,l,X) \leq \mathcal{P}(t,l,\vrelone(X))$, with 
  $\vrelone(X) = \{y\mid\exists x \in X \text{ such that } x\;\vrelone\;y\}$. Similarly, a \emph{ probabilistic bisimulation } is 
  an equivalence relation $\vrelone$ on $\mathcal{S}$ such that $(s,t)\in \vrelone$ implies that for every equivalence class $E$ modulo 
  $\vrelone$, and for every $l\in\mathcal{L}$, $\mathcal{P}(s,l,E) = \mathcal{P}(t,l,E)$.
\end{definition}
Insisting on bisimulations to be equivalence relations has the potential effect of not allowing them to be formed
by just taking unions of other bisimulations. The same can be said about simulations, which are assumed to be partial orders.
Nevertheless:
\begin{proposition}
If $(\vrelone_i)_{i\in I}$ is a collection of probabilistic (bi)simulations, then the reflexive and transitive 
closure of their union, $(\cup_{i \in I} \vrelone_i )^*$, is a (bi)simulation.
\end{proposition}
A nice consequence of the result above is that we can define \emph{probabilistic similarity} (noted $\precsim$)
simply as the relation $\precsim\;=\bigcup\{\vrelone\mid\vrelone \text{ is a probabilistic simulation} \}$. Analogously for the
largest probabilistic bisimulation, that we call \emph{ probabilistic bisimilarity} (noted $\backsim$),
defined as $\backsim\;=\bigcup\{\vrelone\mid\vrelone \text{ is a probabilistic bisimulation} \}$.
\begin{proposition}
Any symmetric probabilistic simulation is a probabilistic bisimulation.
\end{proposition}
A property of probabilistic bisimulation which does not hold in the usual, nondeterministic, setting, is the following:
\begin{proposition}\label{prop:simvsbisim}
$\backsim = \precsim \cap \precsim^{op}$.
\end{proposition}
\subsection{A Concrete Labelled Markov Chain}
Applicative bisimulation will be defined by instantiating Definition~\ref{def:probasim} on
a specific LMC, namely the one modelling evaluation of \PCFLP\ programs.
\begin{definition}
The labelled Markov Chain $\mccbv=(\mccbvstates,\mccbvlabels,\mccbvmatrix)$ is given by:
\begin{varitemize}
\item 
  A set of states $\mccbvstates$ defined as follows:
  $$
  \mccbvstates=\left\{(\termone,\sigma)\mid\termone\in\termty{\sigma}\right\}\uplus\{(\hat{\valone},\sigma)\mid\valone\in\val{\sigma}\},
  $$
  where terms and values are taken modulo $\alpha$-equivalence. A value $\valone$ in the second component
  of $\mccbvstates$ is distinguished from one in the first by using the notation $\hat\valone$.
\item 
  A set of labels $\mccbvlabels$ defined as follows:
  $$
  \values \uplus\types\uplus\mathbb{N}\uplus\mathbb{B}\uplus\{\nilact,\hd,\tl\}
    \uplus\{\fstact,\sndact\}\uplus\{\evact\}, 
  $$ 
  where, again, terms are taken modulo $\alpha$-equivalence, and $\types$ is the set of types.
\item 
  A transition probability matrix $\mccbvmatrix$ such that:
  \begin{varitemize} 
  \item
    For every 
    $\termone\in\termty{\sigma}$, 
    $\mccbvmatrix\left((\termone,\sigma),\sigma,(\termone,\sigma)\right)=1$, and similarly
    for values.
  \item
    For every 
    $\termone\in\termty{\sigma}$, and any value 
    $\valone\in\supp{\diS{\termone}}$, 
    $\mccbvmatrix\left((\termone,\sigma),\evact,(\hat{\valone},\sigma)\right)= \diS{\termone}(\valone)$.
  \item 
    If $\valone\in\val{\sigma}$, then: 
    \begin{varitemize}
    \item 
      If $\sigma = \tau\rightarrow\typthree$, then 
	\begin{varitemize}      
        \item Either there is $\termone$ such that
          $\valone=\abstrac{\termone}$, and for each $\valtwo\in\val{\tau}$, 
          $$
          \mccbvmatrix\left((\hat\valone,\tau\rightarrow \typthree),\valtwo,(\subst{\termone}{\valtwo}{\varone},\typthree)\right)=1.
          $$
        \item Or there is $\termone$ such that $\valone = \fix {\termone}$,  and for each $\valtwo\in\val{\tau}$, 
          $$
          \mccbvmatrix\left((\hat\valone,\tau\rightarrow \typthree),\valtwo,(\subst{\termone}{\fix \termone}{\varone}\, \valtwo,\typthree)\right)=1.
          $$
        \end{varitemize}
      \item 
        If $\sigma = \tau\times\typthree$, then there are
        $\termone,\termtwo$ such that $\valone=\pair{\termone}{\termtwo}$, and 
        we define:
        \begin{align*}
          \mccbvmatrix\left((\hat\valone,\tau\times\typthree),\fstact,(\termone,\tau)\right) &= 1\\
          \mccbvmatrix\left((\hat\valone,\tau\times\typthree),\sndact,(\termtwo,\typthree)\right) &= 1
        \end{align*}
      \item 
        If $\sigma=\ity$, then there is $k\in\mathbb {N}$ such that $\valone=\cnst{k}$ and
        $\mccbvmatrix\left((\hat\valone,\ity),k,(\hat\valone,\ity)\right) = 1$.
      \item 
        If $\sigma=\bty$, then there is $b\in \mathbb{B}$ such that $\valone=\underline{b}$. Then 
        $\mccbvmatrix\left((\hat\valone,\bty),b,(\hat\valone,\bty)\right)=1$.
      \item 
        If $\sigma=[\tau]$, then there are two possible cases:
        \begin{varitemize}
        \item 
          If $\valone=\nil$, then $\mccbvmatrix\left((\hat{\valone},[\tau]),\nilact,(\hat{\valone},[\tau])\right) = 1$.
        \item 
          If $\valone=\lst{\termone}{\termtwo}$, then $\mccbvmatrix\left((\hat{\valone},[\tau]),\hd,(\termone,\tau)\right) = 1$.
          and  $\mccbvmatrix\left((\hat{\valone},[\tau]),\tl,(\termtwo,[\tau])\right) = 1$.
        \end{varitemize} 
      \end{varitemize}
    \end{varitemize} 
    For all $\stateone,l,\statetwo$ such that $\mccbvmatrix(\stateone,l,\statetwo)$ isn't defined above, we have 
    $\mccbvmatrix(\stateone,l,\statetwo)=0$. 
  \end{varitemize} 
\end{definition}
Please observe that if $\valone\in \val \sigma$, both $(\valone, \sigma)$ and $(\hat\valone, \sigma)$ are states 
of the Markov Chain $\mccbv$. For example, the following are all states of $\mccbv$:
\begin{align*}
(\abstrac x, (\ity \rightarrow \ity));\\
(\widehat{\abstrac x},(\ity \rightarrow \ity));\\
(\abstrac x, \left((\ity \rightarrow \ity)\rightarrow (\ity \rightarrow \ity)\right));\\
(\widehat{\abstrac x},\left((\ity \rightarrow \ity)\rightarrow (\ity \rightarrow \ity)\right)).
\end{align*}
A similar Markov Chain was used in \cite{ADLS13} to define bisimilarity for the untyped probabilistic 
$\lambda$-calculus  $\Lambda_\oplus$. We use here in the same way actions which apply a term to a value,
and an action which models term evaluation, namely $\evact$. 
\subsection{The Definition}
We would like to see any simulation (or bisimulation) on the LMC $\mccbv$ as 
a family in $\relct$. As can be easily realized, indeed, any (bi)simulation on $\mccbv$ cannot put
in correspondence states $(\termone,\sigma)$ and $(\termtwo,\tau)$ where $\sigma\neq\tau$, since
each such pair exposes its second component as an action. This then justifies the following:
\begin{definition}
  A \emph{probabilistic applicative simulation} (a PAS in the following), is a family
  $(\relone_\sigma)\in\relct$ such that there exists a probabilistic simulation $\vrelone$
  on the LMC $\mccbv$ such that for every type $\sigma$, and for every
  $\termone,\termtwo\in\termty{\sigma}$ it holds that
  $\rela{\termone}{\relone_\sigma}{\termtwo} \Leftrightarrow \rela{(\termone,\sigma)}{\vrelone}{(\termtwo,\sigma)}$.  
  A \emph{probabilistic applicative bisimulation} (PAB in the following) is defined similarly, requiring
  $\vrelone$ to be a bisimulation rather than a simulation.
\end{definition}
The greatest simulation and the greatest bisimulation on $\mccbv$ are indicated with $\precsim$, and $\backsim$, respectively.
In other words, $\pasv{\sigma}$ is the relation $\{(\termone,\termtwo)\mid(\termone,\sigma) \precsim (\termtwo,\sigma) \}$, while 
$\pabv{\sigma}$ the relation $\{(\termone,\termtwo)\mid(\termone,\sigma) \backsim (\termtwo,\sigma) \}$.

Please notice that $(\pasv{\sigma})$ is the biggest PAS, and that $(\pabv{\sigma})$ is the biggest PAB. We can also see that 
if $(\relone_\sigma)$ is a PAS, and we define the relation $\vrelone$ by: if  $M\relone_\sigma N$ then $(M,\sigma) \vrelone (N, \sigma)$, 
and if $\valone,\valtwo$ are values, and  $\valone\relone_\sigma\valtwo$, then $(\hat\valone,\sigma) \relone (\hat\valtwo, \sigma)$, then 
$\relone$ is a simulation on $\mccbv$. Similarly if we start from an PAB.
\begin{lemma}
For every $(\valone,\valtwo)\in\val{\sigma}\times\val{\sigma}$, 
$(\hat{\valone},\sigma) \precsim (\hat{\valtwo},\sigma)$ if and only if 
$(\valone,\sigma) \precsim (\valtwo,\sigma)$.
\end{lemma}
\begin{proof}
\begin{varitemize}
\item[$\Leftarrow$] 
  If $(\valone,\sigma) \precsim (\valtwo,\sigma)$,  we have: $$\mccbvmatrix \left((\valone,\sigma) ,\evact, X\right)= \diS {\valone} (X) = 
  \left\{
    \begin{array}{ll}
      1&\text{ if }(\hat{\valone},\sigma) \in X \\
      0 &\text{otherwise}
    \end{array} 
  \right.
  $$
  and 
  $$
  \mccbvmatrix \left( (\valtwo,\sigma), \evact ,(\precsim(X)) \right)=
  \left\{
    \begin{array}{ll}
      1&\text{ if }(\hat{\valtwo},\sigma) \in \precsim(X) \\
      0 &\text{otherwise}
    \end{array} 
  \right.
  $$
  As $\precsim$ is a simulation, $\mccbvmatrix \left((\valone,\sigma) ,\evact,X \right)\leq 
  \mccbvmatrix \left( (\valtwo,\sigma),\evact,(\precsim(X))\right)$. We take $X = \{(\hat{\valone},\sigma)\}$, 
  and we can see that we must have $(\hat{\valtwo},\sigma) \in \precsim(X)$, and it follows that 
  $(\hat{\valone},\sigma) \precsim (\hat{\valtwo},\sigma)$.
\item[$\Rightarrow$] 
  Let $\sigma$ be a fixed type.
  Let be ${\vrelone = \{\left((\valone,\sigma),(\valtwo,\sigma)\right)|((\hat{\valone},\sigma)
  \precsim (\hat{\valtwo},\sigma)\}}$. We are going to show: 
  ${\vrelone \subseteq \precsim}$. Let $\vreltwo = \precsim \cup \vrelone$.
  We can see that $\vreltwo$ is a simulation: Let $\stateone$, $\statetwo$ be such that 
 $\stateone \vreltwo \statetwo$. Then
  \begin{varitemize}
  \item 
    Either $\stateone \precsim \statetwo$, and for every action $l$ and subset $X$ of $\mccbvstates$ , 
    $\mccbvmatrix \left(\stateone,l,X \right) \leq \mccbvmatrix \left(\statetwo,l,(\precsim(X))\right) \leq 
    \mccbvmatrix \left(\statetwo, l,(\vreltwo (X))\right)$.
  \item 
    Or there exist $\valone$ and $\valtwo$ such that $\stateone =(\valone,\sigma)$, $\statetwo =(\valtwo,\sigma)$, and 
    ${(\hat{\valone},\sigma) \precsim (\hat{\valtwo},\sigma)}$,and so ${\valone,\valtwo \in \val{\sigma}}$ and, for 
    every action $l$: 
    \begin{varitemize}
    \item 
      either $l = \evact$, and for every ${X \subseteq \mccbvstates}$, $$\mccbvmatrix \left(\stateone, \evact,X \right)=
      \left\{
        \begin{array}{ll}
          1&\text{ if }(\hat{\valone},\sigma) \in X \\
          0 &\text{otherwise}
        \end{array} 
      \right.$$
      If $(\hat{\valone},\sigma)\not \in X$,  ${\mccbvmatrix\left(\stateone,\evact,X\right) = 0 \leq  \mccbvmatrix 
      \left(\statetwo,\evact,(\vreltwo(X))\right)}$. If  $(\hat{\valone},\sigma) \in X$, then ${(\hat{\valtwo},\sigma) 
      \in \precsim(X) \subseteq \vreltwo(X)}$ and so ${\mccbvmatrix\left(\stateone, \evact,X\right)=1=\mccbvmatrix\left(\statetwo, 
      \evact,(\vreltwo(X))\right)}$. 
    \item 
      either $l \neq \evact$, and for every subset $X$ of $\mccbvstates$  : ${\mccbvmatrix \left(\stateone, l,X\right)=
      \mccbvmatrix \left(\statetwo,l, \vreltwo(X) \right) = 0}$
    \end{varitemize}
  \end{varitemize}
  Since $\vreltwo$ is a simulation, $\vreltwo \subseteq \precsim$, and so we have 
  $\{(\valone,\valtwo)|\hat{\valone}\precsim \hat{\valtwo}\} \subseteq \precsim$.  
\end{varitemize}
This concludes the proof.
\end{proof}
Terms having the same semantics need to be bisimilar:
\begin{lemma}\label{lemma:kleene}
Let $(\relone_\sigma)\in\relct$ be defined as follows: 
$\termone\;\relone_\sigma\;\termtwo\Leftrightarrow \termone,\termtwo\in \termty{\sigma}\wedge\diS{\termone}=\diS{\termtwo}$.
Then $(\relone_\sigma)$ is a PAB.
\end{lemma}
\begin{proof}
Let $\vrelone=\bigcup_{\sigma}\left(\{((\termone,\sigma),(\termtwo,\sigma))\mid \termone\relone_\sigma\termtwo\} 
\cup\{((\hat{\valone},\sigma),(\hat{\valone},\sigma))\mid\valone\in\val{\sigma}\}\right)$.
We proceed by showing that $\vrelone$ is a bisimulation. Now:
\begin{varitemize}
\item 
  For every $\sigma$, $\relone_\sigma$ is an equivalence relation, so $\vrelone$ is an equivalence relation too. The equivalence classes of $\vrelone$ are : 
  the $(\{\sigma\} \times E_\sigma)$ when $E_\sigma$ is an equivalence class of $\relone_\sigma$, and the $\{(\hat{\valone},\sigma)\}$ when $\valone \in \val \sigma$. 
\item 
  For every $\stateone, \statetwo \in \mccbvstates$ such that $\stateone \vrelone \statetwo$, for every $E$ equivalence class of $\vrelone$, for all action $l$: $\mccbvmatrix \left( \stateone, l,E\right) = \mccbvmatrix \left(\statetwo,l,E \right)$.
  Indeed, let $\stateone,\statetwo$ be such that $\stateone \vrelone \statetwo$. There are two possibles cases:
  \begin{varitemize}
  \item 
    There are $\sigma$ a type, and $ \termone, \termtwo \in \termty{\sigma}$, such that $\stateone = (\termone,\sigma), \, \statetwo =(\termtwo,\sigma)$, and $\diS \termone = \diS \termtwo$.
    Let $l$ be an action:
    \begin{varitemize}
    \item 
      either $l = \evact$. Then for every $\statethree \in \mccbvstates$, $\mccbvmatrix\left(\stateone, \evact,\statethree \right) >0 \Leftrightarrow (\statethree = (\hat{\valthree},\sigma)\text{ and } \valthree \in \supp{\diS \statethree}$.
      Let $E$ be an equivalence class of $\vrelone$. By construction of $\vrelone$, we can see that:
      \begin{varitemize}
      \item 
        Or $\exists \typtwo$, such that $E = \{\typtwo\}\times E_\typtwo$, and since the element of $E_\typtwo$ are not distinguished values,  
        $\mccbvmatrix \left(\stateone,\evact,E \right) = 0=\mccbvmatrix \left(\statetwo,\evact,E \right)$.
      \item 
        Or $\exists \typtwo \neq \sigma$,$\valone \in \val \typtwo$ such that $E = \{( \hat{\valone},\typtwo) \}$, 
        and  $\mccbvmatrix \left(\stateone,\evact,E \right) = 0=\mccbvmatrix \left(\statetwo,\evact,E \right)$.
      \item 
        Or $\exists \valone \in \val \sigma $, such that $E = \{(\hat{\valone},\sigma) \}$, 
        and $\mccbvmatrix \left(\stateone,\evact,E \right) = \diS \termone (\valone) = \diS \termtwo (\valone) = \mccbvmatrix \left(\statetwo,\evact,E \right)$. 
      \end{varitemize}  
    \item 
      Or $l \neq \evact$, and for all equivalence class $E$ of $\vrelone$:   $\mccbvmatrix \left(\stateone,\evact,E \right) = 0 =\mccbvmatrix \left(\statetwo,\evact,E \right)$.
    \end{varitemize}
  \item 
    $\exists \sigma$, and $\valone \in \val \sigma$ such that $\termone = (\hat{\valone},\sigma) = \termtwo$, and we have: for every $E$ equivalence class of $\vrelone$, 
    for every action $l$, $\mccbvmatrix \left(\stateone,l,E \right) =  \mccbvmatrix \left(\statetwo,l,E)\right)$.   
  \end{varitemize}
\end{varitemize}
\end{proof}
As a consequence of the previous lemma, if $\termone,\termtwo\in\termty{\sigma}$ are such that 
$\diS{\termone}=\diS{\termtwo}$, then $\termone\pabv{\sigma}\termtwo$.
\begin{example}
  For all $\sigma$, $\termone$, $\termtwo$ such that $\emptyset \proves\termone,\termtwo:\sigma$ 
  and $\diS{N}= \emptyset$, we have that $\termone\pasv{\sigma}\termtwo$ implies $\diS{M}=\emptyset$.
  For every terms $\termone,\termtwo$ such that $x:\tau \proves \termone: \sigma$, and $\emptyset \proves \termtwo: \tau$, we have,
  as a consequence of  Lemma \ref{lemma:kleene}, that $(\abstrac\termone)\termtwo\pabv{\sigma}\subst{\termone}{\termtwo}{x}$.
\end{example}
We have just defined applicative (bi)simulation as a family $({\relone}_\sigma)_\sigma$, each 
${\relone}_\sigma$ being a relation on closed terms of type $\sigma$. We can extend it
to a \emph{typed} relation, by the usual open extension:
\begin{definition}
\begin{varenumerate}
\item
  If $\Gamma = \varone_1: \tau_1,\ldots,\varone_n:\tau_n$ is a context, a 
  \emph{$\Gamma$-closure} makes each variable $\varone_i$ to
  correspond to a value $\valone_i\in\val{\tau_i}$ (where $1\leq i \leq n$).
  The set of $\Gamma$-closures is $\ngc{\Gamma}$. For every term $\Gamma\proves\termone:\sigma$
  and for every $\Gamma$-closure $\ccone$, $\termone\ccone$ is the term in $\termty{\sigma}$ obtained
  by substituting the variables in $\Gamma$ with the corresponding values from $\ccone$.
\item
  Let be $\relone=(\relone_\sigma)\in\relct$. We define the \emph{open extension} of 
  $(\relone_\sigma)$ as the typed relation $\relone_\circ=(\reltwo_\sigma^\Gamma)$ where 
  $\reltwo_\sigma^\Gamma\subseteq\termtyg{\sigma}{\Gamma}\times\termtyg{\sigma}{\Gamma}$ is defined
  by stipulating that 
  $\termone\reltwo_\sigma^\Gamma\termtwo$ iff for every $\ccone\in\ngc{\Gamma}$, $(\termone\ccone)\;\relone_\sigma\;(\termtwo\ccone)$.
\end{varenumerate}
\end{definition}
The following proposition say that $\Ev$ is exactly the intersection of $\Rv$ and of the opposite of $\Rv$.
\begin{proposition}\label{prop:appsimvsappbisim} 
  $\Gamma \proves\termone\Ev\termtwo:\sigma$ iff $\Gamma \proves\termone\Rv\termtwo:\sigma$ and $\Gamma \proves \termtwo \Rv \termone: \sigma$.
\end{proposition}
\begin{proof}
It follows from Proposition \ref{prop:simvsbisim}.
\end{proof}
\begin{lemma}
$\Rv$ is a transitive and reflexive typed relation, and
$\Ev$ is a transitive, reflexive and symmetric typed relation. 
\end{lemma}
\begin{proof}
\begin{varitemize}
\item
  Since $\precsim$ is a preorder on $\mccbvstates$, $\pasv{\sigma}$ is a preorder on $\termty{\sigma}$, too. 
  By definition of $\Rv$, we have the thesis.
\item
  Since $\backsim$ is an equivalence relation on $\mccbvstates$, $\pabv{\sigma}$ is an equivalence relation on 
  $\termty{\sigma}$, too. By definition of $\Ev$, we have the thesis.
\end{varitemize}
\end{proof}
\begin{definition}[Simulation Preorder and Bisimulation Equivalence]
The typed relation $\Rv$ is said to be the \emph{simulation preorder}.
The typed relation $\Ev$ is said to be \emph{bisimulation equivalence}.
\end{definition}
\subsection{Bisimulation Equivalence is a Congruence}
In this section, we want to show that $\Ev$ is actually a congruence, and
that $\Rv$ is a precongruence. In view of Proposition \ref{prop:appsimvsappbisim}, 
it is enough to show that the typed relation $\Rv$ is a precongruence, since $\Ev$ is the intersection of $\Rv$ 
and the opposite relation of $\Rv$. The key step consists in showing that $\Rv$ is compatible. This will be 
carried out by the Howe's Method, which is a general method for establishing such congruence properties~\cite{Howe96}.

The main idea of Howe's method consists in defining an auxiliary relation $\Rvh$, such that it is easy to see that 
it is compatible, and then prove that $\Rv=(\Rvh)^+$. 
\begin{definition}
Let $\relone$ be a typed relation. 
We define inductively the typed relation $\relone^H$ 
by the rules of Figure \ref{fig:howerules}.
\end{definition}
\begin{figure}[htpb]
\fbox{
\begin{minipage}{1.10\textwidth}
{\center
\scriptsize
\def\defaultHypSeparation{\hskip .1in}
\AxiomC{$\Gamma,\varone:\sigma\proves\rela{\varone}{\vrelone}{\termone}:\sigma$}
\UnaryInfC{$\Gamma,\varone:\sigma\proves\rela{\varone}{\vrelone^H}{\termone}:\sigma$}
\DisplayProof \quad
\AxiomC{$\Gamma \proves \underline{n} R \termone: \ity$}
\UnaryInfC{$\Gamma \proves \underline{n} R^H \termone: \ity$}
\DisplayProof \quad
\AxiomC{$\Gamma \proves \underline{bv} R \termone: \bty$}
\UnaryInfC{$\Gamma \proves \underline{bv} R^H \termone: \bty$}
\DisplayProof
\quad
\AxiomC{$\Gamma \proves \termone R^H \termtwo: \ity$}
\AxiomC{$\Gamma \proves \termthree R^H \termfour: \ity$}
\AxiomC{$\Gamma \proves \termtwo \, op\, \termfour R \termfive: \gamma_{\mathsf {op}}$}
\TrinaryInfC{$\Gamma \proves \termone\, op\, \termthree R^H \termfive: \gamma$}
\DisplayProof \quad
\AxiomC{$\Gamma,x:A \proves \termone\,  R^H\,  \termtwo: B$}
\AxiomC{$\Gamma \proves \abstrac \termtwo \, R \, \termthree: A \rightarrow B$}
\BinaryInfC{$\Gamma \proves \abstrac \termone\,  R^H \, \termthree: A \rightarrow B$}
\DisplayProof \quad
\AxiomC{$\Gamma,x:A \proves \termone \, R^H \,\termtwo: A$}
\AxiomC{$\Gamma \proves \fix \termtwo\, R\,  \termthree: A $}
\BinaryInfC{$\Gamma \proves \fix \termone\, R^H\, \termthree: A$}
\DisplayProof \quad
\AxiomC{$\Gamma\proves \termone \, R^H \,  \termtwo: A\rightarrow B$}
\AxiomC{$\Gamma\proves \termthree \, R^H \, \termfour: A $}
\AxiomC{$\Gamma \proves \termtwo \termfour \, R \, \termfive: B$}
\TrinaryInfC{$\Gamma \proves \termone \termthree \, R^H\, \termfive:B$}
\DisplayProof \quad
\AxiomC{$\Gamma\proves \termone \, R^H \,  \termthree: A$}
\AxiomC{$\Gamma\proves \termtwo \, R^H \,  \termfour: B $}
\AxiomC{$\Gamma \proves \pair {\termthree} {\termfour} \, R \, \termfive: A \times B$}
\TrinaryInfC{$\Gamma \proves \pair {\termone} {\termtwo} \, R^H\, \termfive:A \times B$}
\DisplayProof \quad
\AxiomC{$\Gamma,x:A \proves P \, R^H \, P': A \times B$}
\AxiomC{$\Gamma \proves \fst {P'}\, R\,  N: A $}
\BinaryInfC{$\Gamma \proves \fst {P} \, \, R^H \, N: A$}
\DisplayProof \quad
\AxiomC{$\Gamma,x:A \proves P \, R^H \, P': A \times B$}
\AxiomC{$\Gamma \proves \snd {P'}\, R \, N: B $}
\BinaryInfC{$\Gamma \proves \snd { P} \,  R^H \, N: B$}
\DisplayProof \quad
\AxiomC{$\Gamma\proves \termone \, R^H \, \termtwo: A$}
\AxiomC{$\Gamma\proves \termthree \, R^H \, \termfour: A $}
\AxiomC{$\Gamma \proves \termtwo \oplus \termfour \, R \, \termfive: A$}
\TrinaryInfC{$\Gamma \proves \termone \oplus \termthree \, R^H\,  \termfive:A$}
\DisplayProof \quad
\AxiomC{$\Gamma\proves \termone \, R^H \, \termfour: \bty$}
\AxiomC{$\Gamma\proves \termtwo \, R^H \, \termfive: A $}
\AxiomC{$\Gamma\proves \termthree \, R^H \, \termsix: A$}
\AxiomC{$\Gamma\proves \ifth {\termfour} {\termfive} {\termsix} \, R\, \termseven: A $}
\QuaternaryInfC{$\Gamma \proves \ifth {\termone} {\termtwo} {\termthree} \, R^H \, \termseven:A$}
\DisplayProof \quad
\AxiomC{$\Gamma\proves H \, R^H \, H': A$}
\AxiomC{$\Gamma\proves T \,R^H \, T': [A] $}
\AxiomC{$\Gamma \proves H'::T' \, R\,  N: [A]$}
\TrinaryInfC{$\Gamma \proves H::T \, R^H \, N:[A]$}
\DisplayProof \quad
\AxiomC{$\Gamma\proves M_1 \, R^H \, M_1': A$}
\AxiomC{$\Gamma\proves L \, R^H \, L': [B] $}
\AxiomC{$\Gamma,h::B,t::[B] \proves M_2 \, R^H \, M_2': A$}
\AxiomC{$\Gamma\proves \caset {L'} {M_1'} {M_2'} \,R\, e_4: A $}
\QuaternaryInfC{$\Gamma \proves \caset {L} {M_1} {M_2} \, R^H\, N:A$}
\DisplayProof
\par}
\end{minipage}}
\caption{Howe's Construction}
\label{fig:howerules}
\end{figure}
We are now going to show, that if the relation $\relone$ we start from satisfies minimal requirements, namely that 
it is reflexive and transitive, then the transitive closure $(\relone^H)^+$ of the Howe's lifting is guaranteed 
to be a precongruence which contains $\relone$. This is a direct consequence of the following results, whose
proofs are standard inductions:
\begin{varitemize}
\item
  Let $\relone$ be a reflexive typed relation.
  Then $\relone^H$ is a compatible.
\item
  Let $\relone$ be transitive. Then : 
\begin{equation}  
\label{NonTrans}
  \left(\Gamma\proves\rela{\termone}{\relone^H}{\termtwo}:\sigma\right)\wedge
  \left(\Gamma\proves\rela{\termtwo}{\relone}{\termthree}:\sigma\right) 
  \Rightarrow \left(\Gamma\proves\rela{\termone}{\relone^H}{\termthree}:\sigma\right) 
  \end{equation}
\item
  If $\relone$ is reflexive and $\Gamma\proves\rela{\termone}{\relone}{\termtwo}:\sigma$, then 
  $\Gamma\proves\rela{\termone}{\relone^H}{\termtwo}:\sigma$.
\item
  If $\relone$ is compatible, then so is $\relone^+$.
\end{varitemize}  
We can now apply the Howe's construction to $\Rv$, since it is clearly reflexive and transitive. The points above
then tell us that $\Rvh$, and $(\Rvh)^+$ are both compatible. What we are left with, then, is proving that $(\Rvh)^+$ is also a simulation.
Let be $\valone\in\val{\sigma}$. For $\valtwo\in\val{\sigma \rightarrow \tau}$, we will note 
${g^V(\valtwo)=\subst{\termtwo}{\valone} x}$ if ${\termone=\abstrac\termtwo}$, and  
${g^V(\valtwo)=\left(\subst{\termtwo}{\fix\termtwo}\varone\right)\,\valone}$ 
if ${\termone=\fix\termtwo}$.
\begin{lemma}
\label{lemma:redv}
For every $\termone,\termtwo$,
$\left(\emptyset \proves M\, \Rv\, N: \sigma \rightarrow \tau\right)$ implies
$\left(\emptyset \proves g^V(M)\, \Rv\, g^V(N): \tau \right)$.
\end{lemma}
\begin{proof}
It follows from the fact that $(\hat M,\sigma \rightarrow \tau) \backsim_v (\hat N,\sigma \rightarrow \tau)$.
\end{proof}
\begin{lemma}
\label{lemma:substitutivev}
$\Rvh$ is value-substitutive: for every typing context $\Gamma$ and for every terms $M,N$
and values $V,W$ such that $\Gamma,x:A \proves M \Rvh N:\sigma$ and 
$\Gamma \proves V \Rvh W:\tau$, it holds that
$\Gamma \proves {\subst{\termone}{\valone}{\varone}}\,\Rvh\,{\subst{\termtwo}{\valtwo}{\varone}}:\sigma$
\end{lemma}
\begin{lemma}
\label{lemma:klPourRv}
For all $\termone,\termtwo$ terms of \PCFLP\,
\begin{varitemize}
\item 
  If $\emptyset \proves  \termone \, \Rv \, \termtwo: {\sigma \rightarrow \tau} $, then 
  for every $X\subseteq\val{\sigma\rightarrow\tau}$, it holds that
  $\diS{\termone}(X)\leq\diS{\termtwo}(\Rv(X))$.
\item 
  If $\left(\emptyset \proves \termone\, \Rv \,\termtwo: {\sigma \times \tau}\right)$,
  then for every  $X\subseteq\val{\sigma \times \tau}$ we have:
  $\diS \termone (X) \leq \diS \termtwo (Y)$, when 
  $Y = \{\pair{\termthree}{\termfour}\mid\exists\pair{\termfive}{\termsix}\in X\wedge\emptyset 
  \proves\termthree\Rv\termfive:\sigma\wedge\emptyset \proves\termfour\Rvh\termsix:\tau\}$.  
\item 
  If $\left(\emptyset \proves \termone\, \Rv  \,\termtwo: {[\sigma]}\right)$  then 
  ${\diS \termone (\mathsf{nil}) \leq \diS \termtwo (\mathsf{nil})}$ and 
  for every $ X \subseteq \val{[\sigma]}$, ${\diS \termone (X) \leq \diS \termtwo (Y)}$
  where $Y$ is the set of those terms $K::L$ such that $\exists H , T$ with $H::T \in X$ and  
  $\emptyset\proves \,\Rv \, K: \sigma$, and $\emptyset \proves  T\, \Rv \,L: {[\sigma]}$.  
\item 
  $\emptyset \proves  \termone\, \Rv\, \termtwo:\ity$ $\Rightarrow$ $\forall k \in \mathbb{N}, 
  \diS \termone (\underline{k}) \leq \diS \termtwo (\underline{k})$.
\item 
  $\emptyset \proves \termone \, \Rv \, \termtwo:\bty$ $\Rightarrow$ $\forall b \in \mathbb{B}, 
  \diS \termone (\underline{b}) \leq \diS \termtwo (\underline{b})$.
\end{varitemize}
\end{lemma}
\begin{proof}
It follows from the definition of $\Rv$.\cqed
\end{proof}
We also need an auxiliary, technical, lemma about probability assignments:
\begin{definition}
$\mathbb{P} = \left(\{ p_i \}_{1 \leq i \leq n},\, \{r_I \}_{I \subseteq \{1,...,n\}} \right)$ is said to be a 
\emph{ probability assignment } if for each ${I \subseteq \{1,..,n\}}$, it holds that $\sum_{i \in I}p_i \leq \sum_{J \cap I \neq \emptyset} r_J$. 
\end{definition}
\begin{lemma}[Disentangling Sets]\label{lemma:disentangling}
Let $P = \left(\{p_i\}_{1\leq i\leq n},\{r_I \}_{I \subseteq \{1,...,n\}}  \right)$ be a probability assignment.
Then for every non-empty $I \subseteq \{1,...,n\}$, and for every $k \in I$, there is $s_{k,I} \in [0,1]$ 
satisfying the following conditions:
\begin{varitemize}
\item for every $I$, it holds that $\sum_{k \in I} s_{k,I} \leq 1$;
\item for every $k \in {1,...,n}$, it holds that $p_k \leq \sum_{k \in I} s_{k,I}\cdot r_I$. 
\end{varitemize}
\end{lemma}  
\begin{proof}
Any probability assignment $\mathbb{P}$ can be seen as a flow network, where nodes are the nonempty subsets of
$\{1,\ldots,n\}$, plus a distinguished source $s$ and a distinguished target $t$. Edges, then, go from $s$ to
each singleton $\{i\}$ (with capacity $p_i$), from every nonempty $I$ to $I\cup\{i\}\subseteq\{1,\ldots,n\}$
whenever $i\notin I$ (with capacity $1$) and from every such $I$ to $t$ (with capacity $r_I$). The
thesis, then, is easily proved equivalent to showing that such a net supports a flow
equals to $\sum p_i$. And, indeed, the fact that this is the maximum flow from $s$ to $t$ can be proved
by way of the Max-Flow Min-Cut Theorem.
\end{proof}
\begin{lemma}[Key Lemma]
For every terms $\termone,\termtwo$,
\begin{varitemize}
\item 
  If $\emptyset \proves  \termone \, \Rvh \, \termtwo: {\sigma \rightarrow \tau} $, then 
  for every $X_1\subseteq\termtyg{\varone:\sigma}{\tau}$ and $X_2\subseteq\termtyg{\varone:\sigma \rightarrow \tau}{\sigma \rightarrow \tau}$, it holds that
  ${\diS{\termone}\left(\abstrac{X_1} \bigcup \fix{X_2}\right)}\leq{\diS{\termtwo}(\Rv\left(\abstrac{Y_1}\, \bigcup \, {\fix{Y_2}}\right))}$, where 
  $Y_1=\{\termthree\in\termtyg{\varone:\sigma}{\tau}\mid \exists\termfour\in X_1.x:\sigma
  \proves \termfour \Rvh \termthree:\tau\}$ and
  $Y_2=\{\termthree\in\termtyg{\varone:\sigma \rightarrow \tau}{\sigma \rightarrow \tau}\mid \exists\termfour\in X_2.x:\sigma \rightarrow \tau
  \proves \termfour \Rvh \termthree:\sigma \rightarrow \tau\}$.
\item 
  If $\emptyset \proves \termone\, \Rvh \,\termtwo:\sigma \times \tau$,
  then for every  $X\subseteq\val{\sigma \times \tau}$ we have:
  $\diS \termone (X) \leq \diS \termtwo (\Rv(Y))$, where 
  $Y = \{\pair{\termthree}{\termfour}\mid\exists\pair{\termfive}{\termsix}\in X\wedge\emptyset 
  \proves\termfive\Rvh\termthree:\sigma\wedge\emptyset \proves\termsix\Rvh\termfour:\tau\}$.  
\item 
  If $\left(\emptyset \proves \termone\, \Rvh  \,\termtwo: {[\sigma]}\right)$  then it holds that
  ${\diS \termone (\mathsf{nil}) \leq \diS \termtwo (\mathsf{nil})}$ and 
  for every $ X \subseteq \val{[\sigma]}$, ${\diS \termone (X) \leq \diS \termtwo (\Rv(Y))}$
  where $Y$ is the set of those $K::L$ such that there are $H,T$ with $H::T\in X$,
  $\emptyset\proves  H \Rvh K: \sigma$, and $\emptyset\proves  T\, \Rvh \,L: {[\sigma]}$.  
\item 
  $\emptyset \proves  \termone\, \Rvh\, \termtwo:\ity$ $\Rightarrow$ $\forall k \in \mathbb{N}, 
  \diS \termone (\underline{k}) \leq \diS \termtwo (\underline{k})$.
\item 
  $\emptyset \proves \termone \, \Rvh \, \termtwo:\bty$ $\Rightarrow$ $\forall b \in \mathbb{B}, 
  \diS \termone (\underline{b}) \leq \diS \termtwo (\underline{b})$.
\end{varitemize}
\end{lemma}
\begin{proof}
We are going to show the following result: 
Let be $\termone \in \termty{}$, and $\distrone$ such that $\evalv \termone {\distrone}$. Then $\distrone$ verifies:
\begin{varitemize}
\item 
  If $\emptyset \proves  \termone \, \Rvh \, \termtwo: {\sigma \rightarrow \tau} $,
  then for every $X_1\subseteq\termtyg{\varone:\sigma}{\tau}$ and $X_2\subseteq\termtyg{\varone:\sigma \rightarrow \tau}{\sigma \rightarrow \tau}$, 
  it holds that $\distrone\left(\abstrac{X_1} \bigcup \fix{X_2}\right) \leq \diS \termtwo (\Rv\left(\abstrac{Y_1}\, \bigcup \, {\fix{Y_2}}\right))$, 
  where $Y_1=\{\termthree\in\termtyg{\varone:\sigma}{\tau}\mid \exists\termfour\in X_1.x:\sigma
  \proves \termfour \Rvh \termthree:\tau\}$ and
  $Y_2=\{\termthree\in\termtyg{\varone:\sigma \rightarrow \tau}{\sigma \rightarrow \tau}\mid \exists\termfour\in X_2.x:\sigma \rightarrow \tau
  \proves \termfour \Rvh \termthree:\sigma \rightarrow \tau\}$.
\item 
  If $\emptyset \proves \termone\, \Rvh \,\termtwo: {\sigma \times \tau}$  then  for every $ X \subseteq \val{\sigma \times \tau}$ it holds that 
  $\distrone (X) \leq \diS \termtwo (\Rv(Y))$,  when 
  $Y = \{\pair{\termthree}{\termfour}\mid\exists\pair{\termfive}{\termsix}\in X\wedge\emptyset\proves\termthree\Rvh\termfive:\sigma\wedge\emptyset
  \proves\termfour\Rvh\termsix:\tau\}$.   
\item 
  If $\emptyset \proves \termone\, \Rvh  \,\termtwo: {[\sigma]}$  then it holds that
  $\distrone (Nil) \leq \diS \termtwo (Nil)$  and that for every  
  $ X \subseteq \val{[\sigma]} \setminus \{Nil \}$, $\distrone (X) \leq \diS N (\Rv(E))$
  when $${E = \{K::L \text{ such that } \exists H , T \text{ with } H::T \in X \text{ and }\\ 
    \left(\emptyset \proves  H \,\Rvh \, K: \sigma\right) \text{ and }\left(\emptyset \proves  T\, \Rvh \,L: {[\sigma]}\right) \}}$$  
\item 
  $\emptyset \proves  \termone\, \Rvh\, \termtwo: {\ity}$  $\Rightarrow$ $\forall k \in \mathbb{N}, \distrone (\underline{k}) \leq \diS \termtwo (\underline{k})$.
\item 
  $\emptyset \proves M \, \Rvh \, N: \bty $  $\Rightarrow$ $\forall b \in \mathbb{B}, \distrone (\underline{b}) \leq \diS \termtwo (\underline{b})  $.
\end{varitemize}
We are going to show this thesis by induction on the structure of the derivation of $\evav M \distrone$.
\begin{varitemize}
\item 
  If the last rule of the derivation is: 
  \begin{prooftree}
    \AxiomC{}	
    \RightLabel{$bv$}
    \UnaryInfC{$ \evav \termone {\emptyset} $ }
  \end{prooftree}
  Then $\distrone$ = $\emptyset$, and for all X, $\distrone (X) = 0$, and it concludes the proof.
\item 
  If the last rule of the derivation is:  
  \begin{prooftree}
    \AxiomC{}	
    \RightLabel{$b_v$}
    \UnaryInfC{$ \evav \valone {\distrv \valone} $ }
  \end{prooftree}
  Then $\termone$ is a value. Some interesting cases:
  \begin{varitemize}
  \item 
    If $\emptyset \proves  \termone\, \Rvh\, \termtwo:\ity $, then $\termone$ is a value of type $\ity$, so it exists $k$ such that $\valone = \termone = \underline{k}$.
    The only possible way to show $\left( \emptyset \proves  \underline{k}\, \Rvh \, N: \ity \right)$ is :
    \begin{prooftree}
      \AxiomC{$\emptyset \proves \underline{k} \Rv \termtwo: \ity$}
      \UnaryInfC{$\emptyset \proves \underline{k} \Rvh \termtwo: \ity$}
    \end{prooftree}
    So $\emptyset \proves \underline{k} \Rv \termtwo: \ity$.
    By Lemma \ref{lemma:klPourRv}, it implies that $\diS \termtwo = \distrv {\underline{k}}$.  
  \item 
    If $\emptyset \proves  \termone\, \Rvh \,\termtwo: \bty $, then $\termone$ is a value of type $\bty$, so it exists 
    $b$ such that $\valone = \termone = \underline{b}$, and it's similar to the previous case.
  \item 
    If $\emptyset \proves \termone\, \Rvh \,\termtwo:{\sigma \rightarrow \tau}$.
    Let $X_1\subseteq\termtyg{\varone:\sigma}{\tau}$ and $X_2\subseteq\termtyg{\varone:\sigma \rightarrow \tau}{\sigma \rightarrow \tau}$. .
    We define $ Y = \left(\abstrac{Y_1}\, \bigcup \, {\fix{Y_2}}\right)$, 
    where  $Y_1=\{\termthree\in\termtyg{\varone:\sigma}{\tau}\mid \exists\termfour\in X_1.x:\sigma \proves \termfour \Rvh \termthree:\tau\}$ 
    and $Y_2=\{\termthree\in\termtyg{\varone:\sigma \rightarrow \tau}{\sigma \rightarrow \tau}\mid 
    \exists\termfour\in X_2.x:\sigma \rightarrow \tau \proves \termfour \Rvh \termthree:\sigma \rightarrow \tau\}$.
    There are two possible cases:
    \begin{varitemize}
    \item 
      Either $\termone = \abstrac{\termthree}$. The only possible way to show $\left(\emptyset \proves \abstrac \termthree \, \Rvh\,  
        \termtwo: {\sigma \rightarrow \tau}\right)$ is : 
      \begin{prooftree}
        \AxiomC{$x:\sigma \proves \termthree \Rvh \termfour: \tau$}
        \AxiomC{$\emptyset \proves \abstrac \termfour \Rv \termtwo: \sigma \rightarrow \tau$}
        \BinaryInfC{$\emptyset \proves \abstrac \termthree \Rvh \termtwo: \sigma \rightarrow \tau$}
      \end{prooftree}
      As $\left(\emptyset \proves \abstrac \termfour\, \Rv \, N: \sigma \rightarrow \tau \right)$, we can see by 
      Lemma \ref{lemma:klPourRv} : $1 =\diS N (\Rv\, \{\abstrac \termfour \})$. Besides,
      $$
      \distrone ( \abstrac {X_1} \bigcup \fix {X_2}) =
      \left\{
        \begin{array}{ll}
          0&\text{ if }  \termthree \not \in X_1 \\
          1 &\text{otherwise}
        \end{array} 
      \right.
      $$
      If $\termthree \not \in X_1$, then 
      $\distrone (\abstrac {X_1} \bigcup \fix {X_2} ) = 0 \leq \diS \termtwo (\Rv(Y))$, and the thesis holds.
      If $ \termthree \in X_1$, then:
      $\distrone (\abstrac {X_1} \bigcup \fix {X_2}) = 1 = 
      \diS \termtwo (\Rv(\{ \abstrac \termfour \})) $.
      To conclude, we need to have: $\left(\Rv (\{\abstrac \termfour\}) \subseteq (\Rv (Y))\right)$. In fact, it is enough to show 
      that $\abstrac \termfour \in Y$, and this is true since $\termfour \in Y_1$.
    \item 
      Or $\termone = \fix {\termthree}$: the proof is similar.
    \end{varitemize}
  \item  
    If $\emptyset \proves  \termone \, \Rvh \, \termtwo: {\sigma \times {\tau}}$, then $\termone = \pair{\termthree_1}{\termthree_2}$.
    We should have: 
    \begin{prooftree}
      \AxiomC{$\emptyset\proves \termthree_1 \Rvh \termfour_1: \sigma$}
      \AxiomC{$\emptyset\proves \termthree_2 \Rvh \termfour_2: \tau $}
      \AxiomC{$\emptyset \proves \pair {\termfour_1} {\termfour_2} \Rv \termtwo:\sigma  \times \tau$}
      \TrinaryInfC{$\emptyset \proves \pair {\termthree_1} {\termthree_2} \Rvh \termtwo:\sigma \times \tau$}
    \end{prooftree}
    And, since $\left(\emptyset \proves \pair {\termfour_1} {\termfour_2} \Rv \termtwo: \sigma \times \tau \right)$, we can see by 
    Lemma \ref{lemma:klPourRv} that:
    $$
    1 =\diS {\termtwo} \left(\{\pair {\termseven_1} {\termseven_2} \text{ s.t. } \emptyset \proves {\termfour_1}\, {\Rv}\, {\termseven_1}: {\sigma}  \text{ and } 
      \emptyset \proves {\termfour_2}\, {\Rv}\, {\termseven_2}:{\tau} \}\right).$$  
    Moreover, by \eqref{NonTrans}, 
    for every $\termsix$ such that $\termsix = \pair {\termfive_1} {\termfive_2} \in \left(\{\pair {\termseven_1} {\termseven_2} \text{ s.t. } \emptyset \proves {\termfour_1}\, {\Rv}\, {\termseven_1}: {\sigma}  \text{ and } \emptyset \proves {\termfour_2}\, {\Rv}\, {\termseven_2}:{\tau} \}\right)$, since $\left(\emptyset \proves \termthree_1 \Rvh \termfour_1:\sigma\right)$, and  $\left( \emptyset \proves {\termfour_1} {\Rv} {\termfive_1} : \sigma \right)$ , we have : $\left(\emptyset \proves {\termthree_1} \, \Rvh\, {\termfive_1}: {\sigma}\right)$.
    Similarly, $\left(\emptyset \proves  {\termthree_2}\, \Rvh \,{\termfive_2}: \tau\right)$.
    So, we have: 
    $$\termsix \in Z = \{\pair {\termseven_1} {\termseven_2} \text{ s.t. } \emptyset \proves  {\termthree_1}\,\Rvh \, {\termseven_1}: \sigma \text{ and } \emptyset \proves {\termthree_2}\, \Rvh \, {\termseven_2}: \tau \}.$$
    And so ${\{\pair {\termseven_1} {\termseven_2} \text{ s.t. } \emptyset \proves {\termfour_1}\, {\Rv}\, {\termseven_1}: {\sigma}  \text{ and } \emptyset \proves {\termfour_2}\, {\Rv}\, {\termseven_2}:{\tau} \}}\subseteq Z$, and consequently : $$
    1 \leq \diS {\termtwo} \left(\{\pair {\termseven_1} {\termseven_2} \text{ s.t. } \emptyset \proves {\termfour_1}\, {\Rv}\, {\termseven_1}: {\sigma}  \text{ and } \emptyset \proves {\termfour_2}\, {\Rv}\, {\termseven_2}:{\tau} \}\right) \leq \diS {\termtwo} (Z).$$ 
    Let be $
    X \subseteq \val{\sigma \times \sigma'}$. If $\pair {\termthree_1}{\termthree_2} \not \in X$,
    $\distrone (X) = 0$, and it concludes the proof.
    If $\pair{\termthree_1}{\termthree_2} \in X$ then:
    $$\distrone (X) =
    1 = \diS {\termtwo} (Z)  \leq \diS {\termtwo} \left(\{\pair {\termseven_1} {\termseven_2} | \exists \pair{\termfive_1}{\termfive_2}\in X \text{s.t.} \emptyset \proves  {\termfive_1} \Rh {\termseven_1}: {\sigma}  \text{ and }\emptyset \proves  {\termfive_2} \Rh {\termseven_2}: {\tau} \}\right)$$, 
    which is the thesis.
  \end{varitemize}
\item 
  If the derivation of $\evalv M D$ is of the following form: 
  \begin{prooftree}
    \AxiomC{$\evalv {M_1} {K}$}
    \AxiomC{$\evalv {M_2} {F}$}	
    \AxiomC{$\{ {\evalv {\subst P v x} {E}} _{P,v}\}_{\abstrac P \in S(\distr K), \, v \in S(\distr F)}$}
    \AxiomC{$\{{\evalv {\subst Q {\fix Q}{x}v } {G}} _{Q,v} \}_{\fix Q \in S(\distr{K}), \, v \in S(\distr F)}$}
    \RightLabel{$b_a$}
    \QuaternaryInfC{$ \evav {M_1M_2}  {\sum_{v \in S(\distr F)} \distr{F} (v) \left( \sum_{\abstrac P \in S(\distr{K})} \distr{K}(\abstrac P).\distr{E}_{P,v}  + \sum_{\fix Q \in S(\distr{K})} \distr{K}(\fix Q) .\distr{G}_{Q,v}\right)}$ }
  \end{prooftree}
  Then $M=M_1 M_2$. Let us suppose that: $\emptyset \proves {M} \, \Rvh \, {N}: B$.  
  The last rule used to prove this 
  should be: 
  \begin{prooftree}
    \AxiomC{$\emptyset\proves M_1 \Rvh M_1': A\rightarrow B$}
    \AxiomC{$\emptyset\proves M_2 \Rvh M_2': A $}
    \AxiomC{$\emptyset \proves M_1' M_2' \R N: B$}
    \TrinaryInfC{$\Gamma \proves M_1 M_2 \Rh N:B$}
  \end{prooftree}
  $S(\distr K)$ is a finite set.
  Let $P_1,..;P_n$ and $Q_1,...Q_m$ such that $S(\distr K) = \abstrac {P_1}...\abstrac {P_1}, \fix {Q_1},... \fix {Q_m}$. Let us consider 
  the sets $\left(K_i = \{\abstrac t\, |\, x:A \proves P_i \, \Rh\,  t:B \}\right)_{1\leq i \leq n}$, and 
  $\left(J_j = \{\fix t\, |\, x:A \rightarrow B \proves Q_i \, \Rh\,  t:A \rightarrow B \}\right)_{1\leq j \leq m}$.
  We have, by induction hypothesis: 
  $\forall I \subseteq \{1,..,n\}, \forall J \subseteq \{1,..,m\}$
  \begin{equation}
    \label{eqDisSet}
    \distr K\left(\bigcup_{i\in I}\{\abstrac P_i\} \cup \bigcup_{j\in J}\{\fix Q_i\}\right)  \leq \diS {M'_1}{\left(\bigcup_{i \in I}\left(\Rv K_i\right)\cup \bigcup_{j \in J}\left(\Rv J_j\right)\right)}
  \end{equation}
  \eqref{eqDisSet} 
  allows us to apply Lemma \ref{lemma:disentangling}: for every $U \in  \bigcup\limits_{1\leq i\leq n}(\Rv K_i) \cup \bigcup\limits_{1\leq j\leq m}(\Rv J_j)  $, 
  there exist $n$ real numbers $r_1^U,...,r_n^U$, and m real numbers $q_1^U,...q_m^U$, such that:
  \begin{align*} 
    \diS{M_1'}( U) &\geq \sum\limits_{1\leq i \leq n} r_i^U +  \sum\limits_{1\leq j \leq m} q_j^U && \forall U \in \bigcup\limits_{1\leq i \leq n} K_i \cup \bigcup\limits_{1\leq j \leq n} J_j  \\
    \distr {K}(\abstrac {P_i}) &\leq \sum\limits_{U \in K_i} r_i^U && \forall \,1 \leq i \leq n  \\
    \distr {K}(\fix {Q_j}) &\leq \sum\limits_{U \in J_j} q_j^U && \forall \,1 \leq j \leq n \\
  \end{align*}
  In the same way, we can apply the induction hypothesis to $M_2$: Let be $S(\distr F) = \{v_1,...,v_l\}$. Let be $X_i = \Rvh(v_i)$. We have by induction hypothesis: $\forall I \subseteq \{1,..,l\}$, $\distr F (\{v_k \,|\, k \in I\}) \leq \diS {M_2} \left( \bigcup_{k \in I} X_k \right)$. So for all $W \in \bigcup_{1\leq k \leq l}X_k$, there exist l real numbers $s_1^W,..,s_l^W$, such that: \begin{align*} 
    \diS{M_2'}( W) &\geq \sum\limits_{1\leq k \leq l} s_k^W
    && \forall W \in \bigcup\limits_{1\leq k \leq l} X_k \\
    \distr {F}(v_k) &\leq \sum\limits_{W \in X_k} s_k^W && \forall \,1 \leq k \leq l  
  \end{align*}
  So we have for every $b \in \val B$: 
  \begin{align*}
    \distrone (b) & = {\sum_{1\leq k \leq l} \distr F (v_k) \left( \sum_{1 \leq i \leq n} \distr{K}(\abstrac P_i).\distr{E}_{P_i,v_k}(b)  + \sum_{1\leq j \leq m} \distr{K}(\fix Q_j) .\distr{G}_{Q_j,v_j}(b)\right)}\\
    & = {\sum_{1\leq k \leq l} \left(\sum\limits_{W \in X_k} s_k^W \right) \left( \sum_{1 \leq i \leq n} \left(\sum\limits_{U \in K_i} r_i^U \right).\distr{E}_{P_i,v_k}(b)  + \sum_{1\leq j \leq m}\left( \sum\limits_{U \in J_j} q_j^U\right) \cdot\distr{G}_{Q_j,v_j}(b)\right)}\\
    & = {\sum_{1\leq k \leq l} \left(\sum\limits_{W \in X_k} s_k^W  \left( \sum_{1 \leq i \leq n} \left(\sum\limits_{U \in K_i} r_i^U \cdot \distr{E}_{P_i,v_k}(b)\right)  + \sum_{1\leq j \leq m}\left( \sum\limits_{U \in J_j} q_j^U \cdot\distr{G}_{Q_j,v_j}(b)\right)\right)\right)}\\
  \end{align*}
  Let be $U \in \bigcup\limits_{1\leq i \leq n} \Rv(K_i) \cup \bigcup\limits_{1\leq j \leq n} \Rv(J_j)$, and $W \in \bigcup\limits_{1\leq k \leq l} X_k$.
  Let be $t_{U,W} =  {\subst T W x}$ if $U = \abstrac T$, and $ {t_{U,W} = \left({\subst T {\fix T} x} W\right)}$ if $U = \fix T$. 
  We can suppose that $B = \tau \times \tau'$ (other cases are similar). 
  Let be $ X \subseteq \val B$. Let be  $E(X) =  \{\pair {z_1} {z_2}\, |\exists {x_1},\, {x_2}, \pair {x_1} {x_2} \in X \text{ and } \left(\emptyset \proves  {x_1}\, \Rvh \, {z_1}: \tau \right), \,\left(\emptyset \proves   {x_2} \,\Rvh \,  {z_2}: {\tau'}\right)\}$.
  We are going to show:
  \begin{varitemize}
  \item[(i)]If $U \in \Rv(K_i)$, and $W \in X_k$, $\distrtwo_{P_i,v_k}(X) \leq \diS {t_{U,W}}(E(X))$.
  \item[(ii)]If $U \in \Rv(J_j)$, and $W \in X_k$, $\distr G_{Q_j,v_k}(X) \leq \diS {t_{U,W}}(E(X))$
  \end{varitemize}
  \ \\
  \begin{varitemize}
  \item [(i)]Let be $U \in \Rv(K_i)$, and $W \in X_k$. Then there exists $S$ such that 
    \begin{equation}
      \label{eqkl1}
      \emptyset \proves \abstrac S \,\Rv\, U: A \rightarrow B
    \end{equation}
    \begin{equation}
      \label{eqkl2}
      x:A \proves P_i \,\Rvh\, S:B
    \end{equation}
    And besides, since $W \in X_k$:
    \begin{equation}
      \label{eqkl3}
      \emptyset \proves v_k \,\Rvh\, W: A
    \end{equation}
    By \eqref{eqkl1}, and Lemma \ref{lemma:redv}, we have: $\emptyset \proves {\subst S W x} \,\Rv \,t_{U,W}$. Moreover, by \eqref{eqkl2}, \eqref{eqkl3} and Lemma \ref{lemma:substitutivev}, we have $\emptyset \proves \subst {P_i} {v_k} x \Rvh  \subst S W x$. \\
    So by \eqref{NonTrans}, it follows:$\emptyset \proves \subst {P_i} {v_k} x \Rvh  t_{U,W}$.
    And, by induction hypothesis applied to $ \subst {P_i} {v_k} x$, it implies that: $\distrtwo_{P_i,v_k}(X) \leq \diS {t_{U,W}}(E(X))$.
    
  \item[(ii)] The proof is similar to (i).
  \end{varitemize}
  
  \begin{align*}
    \distrone (X) &\leq {\sum_{1\leq k \leq l} \left(\sum\limits_{W \in X_k} s_k^W  \left( \sum_{1 \leq i \leq n} \left(\sum\limits_{U \in K_i} r_i^U \cdot \distr{E}_{P_i,v_k}(X)\right)  + \sum_{1\leq j \leq m}\left( \sum\limits_{U \in J_j} q_j^U \cdot\distr{G}_{Q_j,v_j}(X)\right)\right)\right)}\\
    &\leq {\sum_{1\leq k \leq l} \left(\sum\limits_{W \in X_k} s_k^W  \left( \sum_{1 \leq i \leq n} \left(\sum\limits_{U \in K_i} r_i^U \cdot \diS {t_{U,W}}(E(X))\right)  + \sum_{1\leq j \leq m}\left( \sum\limits_{U \in J_j} q_j^U \cdot\diS {t_{U,W}}(E(X))\right)\right)\right)}\\
    &\leq \sum_{W \in \left(\bigcup\limits_{1 \leq k \leq l} X_k\right)} \sum_{U \in \left(\bigcup\limits_{1\leq i \leq n} K_i \cup \bigcup\limits_{1\leq j \leq n} J_j\right)} \left(\sum_{k \text{ s.t. } W \in X_k} s_k^W \right) \cdot \left(\sum_{i \text{ s.t. } U \in K_i} r_i^U + \sum_{j \text{ s.t. } U \in J_j} q_j^U \right) \diS {t_{U,W}}(E(X))\\
    &\leq \sum_{W \in \left(\bigcup\limits_{1 \leq k \leq l} X_k\right)} \sum_{U \in \left(\bigcup\limits_{1\leq i \leq n} K_i \cup \bigcup\limits_{1\leq j \leq n} J_j\right)} \left(\diS {M_2'}(W) \right) \cdot \left( \diS{M_1'}( U) \right) \diS {t_{U,W}}(E(X))\\
    &\leq \diS {M_1'M_2'} (E(X))
  \end{align*}
\end{varitemize}
\end{proof}

A consequence of the Key Lemma, then, is that $(\Rvh)^+$ is an applicative bisimulation, thus included in the largest one, namely $\Rv$. Since the
latter is itself included in $\Rvh$, we obtain that $\Rv=(\Rvh)^+$. But $(\Rvh)^+$ is a precongruence, and we get the main result of this section:
\begin{theorem}[Soundness]\label{theo:soundness}
The typed relation $\Rv$ is a precongruence relation included in $\leq$. Analogously, $\Ev$ is a congruence
relation included in $\equiv$.
\end{theorem}
\subsection{Back to Our Examples}
We now have all the necessary tools to prove that the example programs from Section~\ref{sect:motexa}
are indeed context equivalent. As an example, let us consider again the following terms:
\begin{align*}
  \expc_\mathit{FST}&=\abstr{\varone}{\abstr{\vartwo}{\enc\;\varone\;\gen}}:\arr{\bty}{\arr{\bty}{\bty}};\\
  \expc_\mathit{SND}&=\abstr{\varone}{\abstr{\vartwo}{\enc\;\vartwo\;\gen}}:\arr{\bty}{\arr{\bty}{\bty}}.
\end{align*}
One can define the relations $\relone_{\bty},\relone_{\arr{\bty}{\bty}},\relone_{\arr{\bty}{\arr{\bty}{\bty}}}$
by stipulating that $\relone_{\sigma}=X_\sigma\times X_\sigma\cup\mathit{ID}_\sigma$
where
{\footnotesize
\begin{align*}
  X_\bty&=\{(\enc\;\underline\ttrue\;\gen),(\enc\;\underline\ffalse\;\gen)\};\\
  X_{\arr{\bty}{\bty}}&=\{(\abstr{\vartwo}{\enc\;\vartwo\;\gen}),(\abstr{\vartwo}{\enc\;\underline\ttrue\;\gen}),(\abstr{\vartwo}{\enc\;\underline\ffalse\;\gen})\};\\
  X_{\arr{\bty}{\arr{\bty}{\bty}}}&=\{\expc_\mathit{FST},\expc_\mathit{SND}\};
\end{align*}}
and for every type $\sigma$, $\mathit{ID}_\sigma$ is the identity on $\termty{\sigma}$. When $\sigma$ is
not one of the types above, $\relone_\sigma$ can be set to be just $\mathit{ID}_{\sigma}$. This way, the
family $(\relone_\sigma)$ can be seen as a relation $\vrelone$ on the state space of $\mccbv$ (since any state in
the form $(\hat\valone,\sigma)$ can be treated as $(\valone,\sigma)$). But $\vrelone$ is easily seen
to be a bisimulation. Indeed:
\begin{varitemize}
\item
  All pairs of terms in $\relone_{\bty}$ have the same semantics, since 
  $\diS{\enc\;\underline\ttrue\;\gen}$ and $\diS{\enc\;\underline\ffalse\;\gen}$ are both
  the uniform distribution on the set of boolean values.
\item
  The elements of $X_{\arr{\bty}{\bty}}$ are values, and if we
  apply any two of them to a fixed boolean value, we end up with two terms
  $\relone_{\bty}$ puts in relation.
\item
  Similarly for $X_{\arr{\bty}{\arr{\bty}{\bty}}}$: applying any two elements
  of it to a boolean value yields two elements which are put in relations
  by $X_{\arr{\bty}{\bty}}$.
\end{varitemize}
Being an applicative bisimulation, $(\relone_\sigma)_\sigma$ is included in $\sim$.
And, by Theorem~\ref{theo:soundness}, we can conclude that $\expc_\mathit{FST}\equiv\expc_\mathit{SND}$.
Analogously, one can verify that $\expc\equiv\rand$.
\section{Full Abstraction}
Theorem~\ref{theo:soundness} tells us that applicative bisimilarity is a sound way to prove
that certain terms are context equivalent. Moreover, applicative bisimilarity is a congruence, and
can then be applied in any context yielding bisimilar terms. In this section, we ask ourselves \emph{how close}
bisimilarity and context equivalence really are. Is it that the two coincide?
\subsection{LMPs, Bisimulation, and Testing}
The concept of probabilistic bisimulation has been generalized to the continuous case by Edalat, Desharnais and Panangaden, more
than ten years ago~\cite{DEP02}. Similarity and bisimilarity as defined in the aforementioned paper were later shown to exactly
correspond to appropriate, and relatively simple, notions of \emph{testing}~\cite{vBMOW05}. We will
make essential use of this characterization when proving that context equivalence is included in bisimulation. And this section
is devoted to giving a brief but necessary introduction to the relevant theory. For more details, please refer to~\cite{vBMOW05} and
to~\cite{EV}.

In the rest of this section, $\actset$ is a fixed set of labels. The first step consists in giving a
generalization of LMCs in which the set of states is not restricted to be countable:
\begin{definition}
A \emph{labelled Markov process} (LMP in the following) is a triple $\mpone=(\statone,\fieldone,\mu)$, consisting of a set $\statone$ of states, a 
$\sigma$-field $\fieldone$ on $\statone$, and a transition probability function $\mu:\statone\times\actset\times\Sigma \rightarrow [0,1]$, such that: 
\begin{varitemize}
\item 
  for all $x \in \statone$, and $a \in Act$, the naturally defined function $\mu_{x,a}(\cdot):\Sigma \rightarrow [0,1]$ is a subprobability measure;
\item
  for all $a \in Act$, and $A \in \Sigma$, the naturally defined function $\mu_{(\cdot),a}(A): \statone \rightarrow [0,1]$ is measurable.
\end{varitemize}
\end{definition} 
The notion of (bi)simulation can be smoothly generalized to the continuous case:
\begin{definition}
\label{defBi}
Let $(\statone,\fieldone,\mu)$ be a LMP, and let $\vrelone$ be a reflexive relation on $\statone$. 
We say that $\vrelone$ is a \emph{simulation} if it satisfies condition \ref{point:bislmpone} below, and we say that $\vrelone$ is a \emph{ bisimulation } 
if it satisfies both conditions \ref{point:bislmpone} and \ref{point:bislmptwo}:
\begin{varenumerate}
\item\label{point:bislmpone}
  If $\rela{x}{R}{y}$, then for every $a\in\actset$ and for every $A\in\Sigma$ such that $A=R(A)$, it holds that $\mu_{x,a}(A)\leq \mu_{y,a}(A)$.
\item\label{point:bislmptwo}
  If $\rela{x}{R}{y}$, then for every $a\in\actset$ and for every $A\in\Sigma$, $\mu_{x,a}(\statone) = \mu_{y,a}(\statone)$.
\end{varenumerate}
We say that two states are \emph{bisimilar} if they are related by some bisimulation.
\end{definition}
\begin{lemma}
\label{ProcessLargestBisim}
Let $\langle \statone\,,\Sigma\,,\mu\rangle$ be a labelled Markov process.
\begin{varitemize}
\item 
  There is a largest bisimulation on $(\statone,\fieldone,\mu)$ which is an equivalence relation.
\item 
  For an equivalence relation $\vrelone$, the two criteria in Definition \ref{defBi} can be compressed into the following condition:
  $x\vrelone y \Rightarrow (\forall a \in\actset)(\forall A \in \Sigma)({A = R(A) \Rightarrow \mu_{x,a}(A) = \mu_{y,a}(A)})$.
\end{varitemize}
\end{lemma}
We will soon see that there is a natural way to turn any LMC into a LMP, in such a way that (bi)similarity stays the same. Before doing
so, however, let us introduce the notion of a \emph{test}:
\begin{definition}\label{def:testlan}
The \emph{test language} $\testlan$ is given by the grammar
$\testone::=\omega\midd\actone\cdot\testone\midd\seq{\testone,\testone}$,
where $\actone\in\actset$.
\end{definition}
Please observe that tests are \emph{finite} objects, and that there isn't any disjunctive nor any negative test in $\testlan$.
Intuitively, $\omega$ is the test which always succeeds, while $\seq{\testone,\testtwo}$ corresponds to making 
two copies of the underlying state, testing them independently according to $\testone$ and $\testtwo$ and succeeding iff \emph{both}
tests succeed. The test $\actone\cdot\testone$ consists in performing the action $\actone$, and in case of success perform the test $\testone$.
This can be formalized as follows: 
\begin{definition}
\label{test}
Given a labelled Markov Process $\mpone=(\statone,\fieldone,\mu)$, we define an indexed family 
$\{\prsucc{\mpone}{\cdot}{\testone}\}_{\testone\in\testlan}$ (such that
$\prsucc{\mpone}{\cdot}{\testone}:\statone\rightarrow\RR$) by induction on the
structure of $\testone$:
{\footnotesize
$$
  \prsucc{\mpone}{x}{\omega}=1;\qquad
  \prsucc{\mpone}{x}{\actone\cdot\testone}=\int\prsucc{\mpone}{\cdot}{\testone}d\mu_{x,\actone};\qquad
  \prsucc{\mpone}{x}{\seq{\testone,\testtwo}}=\prsucc{\mpone}{x}{\testone}\cdot\prsucc{\mpone}{x}{\testtwo}. 
$$}
\end{definition}
From our point of view, the key result is the following one:
\begin{theorem}[\cite{vBMOW05}]\label{theo:bisimtestproc}
Let $\mpone=(\statone,\fieldone,\mu)$ be a LMP. Then $x,y\in\statone$ are bisimilar iff
$\prsucc{\mpone}{x}{\testone}=\prsucc{\mpone}{y}{\testone}$ for every test $\testone\in\testlan$. 
\end{theorem}
\subsection{From LMPs to LMCs}\label{sect:fromlmptolmc}
We are now going to adapt Theorem~\ref{theo:bisimtestproc} to LMCs, thus getting an analogous
characterization of probabilistic bisimilarity for them.

Let $\mcone=(\statone,\actset,\matrixone)$ be a LMC. The function $\mu_\mcone:\statone\times\actset\times\powset{\statone}\rightarrow [0,1]$
is defined by $\mu_\mcone(\stone,\actone,\setone)=\sum_{x\in\setone}\matrixone(s,\actone,x)$. This construction allows us to
see any LMC as a LMP:
\begin{lemma}
  Let $\mcone=(\statone,\actset,\matrixone)$ be a LMC. Then 
  $(\statone,\powset{\statone},\mu_\mcone)$ is a LMP, that we denote as $\lmp{\mcone}$.
\end{lemma}
\begin{proof}
\begin{itemize}
\item $\powset{\statone}$ is a $\sigma$-field (non-empty, closed under complementation and countable unions).
\item $\mu$ verifies :
\begin{itemize}
\item for every $\stone \in \statone$, and $\actone \in \actset$, $\mu_{\stone,\actone}$ is a sub-probability measure, since : 
\begin{itemize}
\item $\mu_{\stone,\actone}(\emptyset) = \sum_{x\in\emptyset}\matrixone(s,\actone,x) = 0$
\item $\mu_{\stone,\actone}({\statone}) = \sum_{x\in\statone}\matrixone(s,\actone,x) \leq 1$ since $\matrixone$ is a probability matrix.
\item For all countable collection of pairwise disjoints $A_n \in \powset{\statone}$,
\begin{align*} 
\mu_{\stone,\actone}\left(\bigcup_{n}A_n\right) &= \sum_{x\in(\bigcup_{n}A_n)}\matrixone(s,\actone,x)\\
& = \sum_{n}\left(\sum_{x \in A_n}\matrixone(s,\actone,x)\right) \\
& = \sum_n \mu_{\stone,\actone}(A_n)
\end{align*}
\end{itemize}
\item for every $\actone \in \actset$, and $A \in \powset{\statone}$, $\mu_{-,\actone}(A): \statone \rightarrow [0,1]$ is measurable since : $\forall I \subseteq [0,1]$, ${\mu_{-,\actone}(A)}^{-1} \in \powset{\statone}$. 
\end{itemize}
\end{itemize}
\end{proof}
But how about bisimulation? Do we get the same notion of equivalence this way? The answer
is positive:
\begin{lemma}
\label{BisimChainProcess}
  Let $\mcone=(\statone,\actset,\matrixone)$ be a LMC,
  and let $R$ be an equivalence relation over $\statone$. Then $R$ is 
  a bisimulation with respect to $\mcone$ if and only if $R$ is a bisimulation 
  with respect to $\lmp{\mcone}$. Moreover, two states are bisimilar with respect 
  to $\mcone$ iff they are bisimilar with respect to $\lmp{\mcone}$.
\end{lemma}
\begin{proof}
Let $\relone$ be an equivalence relation over $\statone$. 
\begin{itemize}
\item[$\Rightarrow$] We suppose that $\relone$ is a bisimulation with respect to the Markov chain $(\statone,\actset,\matrixone)$. 
By Lemma \ref{ProcessLargestBisim}, it is enough to show that : $$x\vrelone y \Rightarrow (\forall a \in\actset)(\forall A \in \powset{\statone})({A = R(A) \Rightarrow \mu_{x,a}(A) = \mu_{y,a}(A)}).$$
Let $x,y$ be such that $x \relone y$.
 For every $\actone \in \actset$ and $A \in \powset{\statone}$ such that $A = \relone(A)$, we have : A is a $\relone$-equivalence class. So (since $\relone$  is a bisimulation with respect to the Markov chain), $\mu_{x,\actone}(A) = \sum_{\stone \in A}\matrixone(x,\actone,\stone) = \sum_{\stone \in A}\matrixone(y,\actone,\stone)= \mu_{y,\actone}(A)$.

\item[$\Leftarrow$] We suppose that $\relone$ is a bisimulation with respect to the Markov process $(\statone,\powset{\statone},\mu)$. Then $\relone$ is a bisimulation with respect to the Markov chain $(\statone,\actset,\matrixone)$, since
\begin{itemize}
\item $\relone$ is an equivalence relation.
\item Let $x,y$ be such that $x \relone y$. For every $\actone \in \actset$, and E a $\relone$-equivalence class, 
\begin{align*}
\sum_{\stone \in E}\matrixone(x,\actone,\stone) &= \mu_{x,\actone}(E) \\
&= \mu_{y,\actone}(E) &\text{by Lemma \ref{ProcessLargestBisim}}\\
&= \sum_{\stone \in E}\matrixone(y,\actone,\stone)
\end{align*}
\end{itemize}
\end{itemize}
About the second statement:
\begin{itemize}
\item[$\Rightarrow$]
Let $x,y$ be two states which are bisimilar with respect to the Markov chain $(\statone,\actset,\matrixone)$. Then there is a $\relone$ a bisimulation with respect to the Markov chain $(\statone,\actset,\matrixone)$ such that $x \relone y$. It follows from Lemma \ref{BisimChainProcess} that $\relone$ is a bisimulation with respect to the Markov Process $(\statone,\powset{\statone},\mu)$. So, $x$ and $y$ are bisimilar with respect to the Markov process $(\statone,\powset{\statone},\mu)$.    
  \item[$\Leftarrow$] Let $x,y$ be two states which are bisimilar with respect to the Markov 
  Process $(\statone,\powset{\statone},\mu)$. Then (by Lemma \ref{ProcessLargestBisim}) we can consider $\relone$ the largest bisimulation (with respect to the Markov Process), and we know that $\relone$ is an equivalence relation. We have : $x \relone y$. It follows from Lemma \ref{BisimChainProcess} that $\relone$ is a bisimulation with respect to the Markov chain, and so $x$ and $y$ are bisimilar with respect to the Markov Chain.
  \end{itemize}
\end{proof}
Let $\mcone=(\statone,\actset,\matrixone)$ be a LMC.
We define an indexed family $\{\prsucc{\mcone}{\cdot}{\testone}\}_{\testone\in\testlan}$ by 
$\prsucc{\mcone}{x}{\testone}=\prsucc{\lmp{\mcone}}{x}{\testone}$, the latter
being the function from Definition~\ref{test} applied to the Markov process $\lmp{\mcone}$.
As a consequence of the previous results in this section, we get that:
\begin{theorem}\label{theo:testbisim}
  Let $\mcone=(\statone,\actset,\matrixone)$ be a LMC.
  Then two states $x,y\in\statone$ are bisimilar if and only if 
  for all tests $\testone\in\testlan$, 
  $\prsucc{\mcone}{x}{\testone}=\prsucc{\mcone}{y}{\testone}$. 
\end{theorem}
\begin{proof}
\begin{align*}
x \text{ and } y \text{ are bisimilar } &\Leftrightarrow x \text{ and } y \text{ are bisimilar with respect to the Markov Process }\lmp{\mcone} = (\statone,\powset{\statone},\mu) \\ 
&\Leftrightarrow \forall \testone \in \testlan, \,
\prsucc{\lmp{\mcone}}{x}{\testone} = \prsucc{\lmp{\mcone}}{y}{\testone} \, \, \text{ by Theorem \ref{theo:bisimtestproc}} \\
&\Leftrightarrow \forall \testone \in \testlan,\, \prsucc{\mcone}{x}{\testone}=\prsucc{\mcone}{y}{\testone}
\end{align*} 
\end{proof}
The last result derives appropriate expressions for the $\prsucc{\mcone}{\cdot}{\cdot}$, which will be
extremely useful in the next section: 
\begin{proposition}
Let $\mcone=(\statone,\actset,\matrixone)$ be a LMC. For all $x\in\statone$, and $\testone\in\testlan$, we have: 
{\scriptsize
$$
\prsucc{\mcone}{x}{\omega}=1;\qquad
\prsucc{\mcone}{x}{\actone\cdot\testone}=\sum_{\stone\in\statone}\matrixone(x,\actone,\stone)\cdot\prsucc{\mcone}{\stone}{\testone};\qquad 
\prsucc{\mcone}{x}{\seq{\testone,\testtwo}}=\prsucc{\mcone}{x}{\testone}\cdot\prsucc{\mcone}{x}{\testtwo}. 
$$}
\end{proposition}
\subsection{Every Test has an Equivalent Context}
We are going to consider the labelled Markov Chain $\mccbv$ defined previously. We know that two programs $\termone$ 
and $\termtwo$ in $\termty\sigma$ are bisimilar if and only if the states $(\termone,\sigma)$ and $(\termtwo,\sigma)$ 
have exactly the same probability to succeed for the tests in $\testlan$, measured according to $\prsucc{\mcone}{\cdot}{\cdot}$. 
Proving that context equivalence is included in bisimulation boils down to show that if $\termone$ and $\termtwo$ have 
exactly the same convergence probability for all contexts, then they have exactly the same success probability 
for all tests. Or, more precisely, that for a given test $\testone$, and a given type $\sigma$, there exists 
a context $\ctxone$, such that for all term $\termone$ of type $\sigma$, the success probability of $\testone$ on 
$(\termone,\sigma)$ is \emph{exactly} the convergence probability of $\act\ctxone\termone$:
$$
\prsucc{\mccbv}{(\termone,\sigma)}{\testone}=\sum\diS{\act{\ctxone}{\termone}}.
$$
However, we should take into account states in the form $(\hat{\valone},\sigma) \in \mccbvstates$, where $\valone$ is a value. 
The formalisation of the just described idea is the following Lemma: 
\begin{lemma}\label{lemma:testcontext}
Let $\sigma$ be a type, and $\testone$ a test.
Then there are contexts $\ctxone_\testone^\sigma$, and ${\ctxtwo_\testone^\sigma}$  such that 
${{\emptyset} \proves {\tycon {\ctxone_\testone^\sigma} {\emptyset} {\sigma} {\bty}}}$, 
${\emptyset \proves \tycon {\ctxtwo_\testone^\sigma} {\emptyset} {\sigma} {\bty} }$, and
for every $\termone \in \termty{\sigma}$ and every $\valone \in \val{\sigma}$,
it holds that
$$
\prsucc{\mccbv}{(\termone,\sigma)}{\testone}=\sum{\diS{\act {\ctxone_\testone^\sigma}{\termone}}};\qquad
\prsucc{\mccbv}{(\hat{\valone},\sigma)}{\testone}=\sum{\diS{\act {\ctxtwo_\testone^\sigma}{\valone}}}.
$$
\end{lemma}
\begin{proof}
  We are going to show the thesis by induction on $\testone$.
  \begin{varitemize}
  \item 
    if $\testone = \omega$, then $\forall \sigma$, we define $\ctxone_{\omega}^\sigma = (\abstrac {\underline{\ttrue}}) (\lambda z.\emptycon)$, and 
    ${\ctxtwo_{\omega}^\sigma} = (\abstrac {\underline{\ttrue}}) (\lambda z.\emptycon)$. And we have: 
    $$\forall \sigma,\, \forall \termone \in \termty{\sigma}, \prsucc{\mccbv}{(\termone,\sigma)}{\omega} = 1 = \sum\left({\diS{\act {\ctxone_\omega^\sigma} {\termone}}}\right)$$
    and 
    $$
    \forall \sigma,\, \forall \valone \in \val{\sigma}, 
\prsucc{\mccbv}{(\hat{\valone},\sigma)}{\omega}=1 = \sum\left({\diS{\act {\ctxtwo_\omega^\sigma}{\valone}}}\right).
    $$
  \item 
    If $\testone = \seq{\testtwo_1,...,\testtwo_n}$.
    Let $\sigma$ be a type.
    By induction hypothesis, for all $1\leq i \leq n$, there exist $\ctxone_{\testtwo_i}^\sigma$ and ${\ctxtwo_{\testtwo_i}^\sigma}$, such that 
    \begin{align*}
     &\emptyset \proves {\tycon {\ctxone_{\testtwo_i}^\sigma}{\sigma}{\emptyset} {\bty}} \text{ and } \emptyset \proves {\tycon{\ctxtwo_{\testtwo_i}^\sigma}{\sigma}{\emptyset}{\bty}} \\
      &\forall \termone \in \termty{\sigma}, 
      \prsucc{\mccbv}{(\termone,\sigma)}{\testtwo_i} = \sum\left({\diS{\act {\ctxone_{\testtwo_i}^\sigma} \termone}}\right)\\
      &\forall \valone \in \val{\sigma}, 
\prsucc{\mccbv}{(\hat{\valone},\sigma)}{\testtwo_i}  = \sum\left({\diS{\act {\ctxtwo_{\testtwo_i}^\sigma} \valone}}\right)\\
    \end{align*}
    We define: 
    \begin{align*}
      \ctxone_\testone^\sigma = (\abstrac T)(\lambda z.\emptycon)\\
      {\ctxtwo_\testone^\sigma} = (\abstrac T^*)(\lambda z.\emptycon)
    \end{align*}
    where:
    {\small
      \begin{align*}
        T =&\ifth {\left( (\lambda y.{\underline{\ttrue}}) (\ctxone_{\testtwo_1}^\sigma [x\, \underline{0}]) \right) \\
        &\;\;\;\;\;\;\;\; } { \left(
        {\ifth {(({\lambda y.{\underline{\ttrue}}})( \ctxone_{\testtwo_2}^\sigma [x\, \underline{0}]))  }
         {...(\ifth {(\lambda y.{\underline{\ttrue}}) (\ctxone_{\testtwo_n}^\sigma [x \,\underline{0}])} {\underline{\ttrue}} {\underline{\ttrue}})...}
          {\underline{\ttrue}} }\right) \\ &} {\underline{\ttrue}}\\
              T^* =& \text{ if } \left( (\lambda y.{\underline{\ttrue}}) ({D_{\testtwo_1}^\sigma} [x\, \underline{0}]) \right) \text{then } \left(\text{if} (\lambda y.{\underline{\ttrue}})\left( {D_{\testtwo_2}^\sigma }[x\, \underline{0}] \right)...(\text{if} (\lambda y.{\underline{\ttrue}}) ({D_{\testtwo_n}^\sigma} [x\, \underline{0}]) \text{ then } {\underline{\ttrue}} \text{ else } {\underline{\ttrue}})...\text{else } {\underline{\ttrue}} \right) \text{else } {\underline{\ttrue}}
      \end{align*}
    }.
    
    Let be $\termone \in \termty \sigma$. We have: 
    \begin{align*}
      \sum\left({\diS{\act {\ctxone_\testone^\sigma}{\termone}}}\right) & = \prod_{1\leq i \leq n}\sum\left({\diS{\act {\ctxone_{\testtwo_i}^\sigma}{(\lambda z.\termone){\underline{0}}}}}\right) \\
      &=\prod_{1\leq i \leq n}\sum\left({\diS{\act{\ctxone_{\testtwo_i}^\sigma}{\termone}}}\right) \text{ since } \forall \termone \text{ closed term }, \termone \backsim_{ctx} (\lambda z.\termone){\underline 0}\\
      &=\prsucc{\mccbv}{(\termone,\sigma)}{\testtwo_1}\cdot ...\cdot \prsucc{\mccbv}{(\termone,\sigma)}{\testtwo_n} \\
      &=\prsucc{\mccbv}{(\termone,\sigma)}{\seq{\testtwo_1,...,\testtwo_n}}
    \end{align*}
    And similarly we have, for every $\valone \in \val{\sigma}$:
    \begin{equation*} 
      \sum\left({\diS{{\act {\ctxtwo_\testone^\sigma}{\valone}}}}\right)  =\prsucc{\mccbv}{(\hat \valone,\sigma)}{\seq{\testtwo_1,...,\testtwo_n}}
    \end{equation*}
    
  \item 
    if $\testone = \actone\cdot \testtwo$
    \begin{varitemize}
    \item if $\actone = \evact$, we define: 
      \begin{equation*}
        {\ctxtwo_{\testone}^\sigma} = (\abstrac {\emptycon})\Omega
      \end{equation*}
      and
      \begin{equation*}
        \ctxone_{\testone}^\sigma = \left(\abstrac ({\act{\ctxtwo_{\testtwo}^\sigma}{\varone}})\right)(\emptycon)
      \end{equation*}
      And we have $\forall \termone \in \termty{\sigma}$: 
      \begin{align*}
        \sum\left({\diS{\act{C_t^\sigma}{\termone}}}\right) & = 
        \sum_{\valone \in \val{\sigma}} \diS{\termone}(\valone)\cdot \sum\left({\diS{{{D_{\testtwo}^\sigma}[\valone]}}}\right)\\
        & = \sum_{\valone \in \val{\sigma}} \diS{\termone}(\valone)\cdot \prsucc{(\hat \valone,\sigma)}{\testtwo} \\
        &= \sum_{e \in \mccbvstates} \mccbvmatrix((\termone,\sigma),\evact,y)\cdot \prsucc{\mccbv}{e}{\testtwo} \\
        &=\prsucc{\mccbv}{(\termone,\sigma)}{\testtwo}
      \end{align*}
      
    \item if $\actone = \valone$, with $\valone \in {\values}$.
      we define $\ctxone_\testone^\sigma = (\abstrac \emptycon)\Omega$.
      \begin{varitemize}
      \item if $\sigma = \tau_1 \rightarrow \tau_2$, and $\valone \in \val {\tau_1}$, then we define:
        \begin{equation*} 
          {\ctxtwo_\testone^{\tau_1 \rightarrow \tau_2}} = {\ctxone_{\testtwo}^{\tau_2}}[\emptycon \valone]
        \end{equation*}
      \item otherwise, we define: ${\ctxtwo_\testone^\sigma} = (\abstrac {\emptycon})\Omega$.
      \end{varitemize}
\item 
  if $\actone = \fstact$:
  we define $C_t^\sigma = (\abstrac \emptycon)\Omega$.
  \begin{varitemize}
  \item if $\sigma = \tau_1 \times \tau_2$ then we define:
    \begin{equation*} 
      {D_t^{\tau_1 \times \tau_2}} = {C_{\testtwo}^{\tau_1}}[\fst \emptycon]
    \end{equation*}
  \item otherwise, we define: ${D_t^\sigma}=(\abstrac \emptycon)\Omega$.
  \end{varitemize} 
\item 
  if $\actone = \sndact$: similar to the previous case.
\item 
  if $\actone = \hd$:
  we define $C_t^\sigma = (\abstrac \emptycon)\Omega$.
  \begin{varitemize}
  \item if $\sigma = [\tau]$ then we define:
    \begin{equation*} 
      {D_t^{[\tau]}} = {C_{\testtwo}^{\tau}}[\caset{\emptycon} {\Omega} {h}]
    \end{equation*}
  \item otherwise, we define: ${D_t^\sigma}=(\abstrac \emptycon)\Omega$.
  \end{varitemize} 
\item 
  if $\actone = \tl$:
  we define $C_t^\sigma = (\abstrac \emptycon)\Omega$.
  \begin{varitemize}
  \item if $\sigma = [\tau]$ then we define:
    \begin{equation*} 
      {D_t^{[\tau]}} = {C_{\testtwo}^{[\tau]}}[\caset{\emptycon} {\Omega} {t}]
    \end{equation*}
  \item otherwise, we define: ${D_t^\sigma}=(\abstrac \emptycon)\Omega$.
  \end{varitemize} 
\item 
  if $\actone = \nilact$:
  we define $C_t^\sigma = (\abstrac \emptycon)\Omega$.
  \begin{varitemize}
  \item if $\sigma = [\tau]$ then we define:
    \begin{equation*} 
      {D_t^{[\tau]}} = {C_{\testtwo}^{\gamma}}[\caset{\emptycon} {C_{\testtwo}^\gamma[\tilde div]} {\Omega}]
    \end{equation*}
  \item otherwise, we define: ${D_t^\sigma}=(\abstrac \emptycon)\Omega$.
  \end{varitemize}
\item 
  if $\actone = k$, with $k \in \mathbb{N}$.
  \begin{varitemize}
  \item if $\sigma = \ity$ we define: 
    ${D_t^{\ity}} =\ifth{(\emptycon = \underline k)}{\underline{\ttrue}}{\Omega}$.
  \item otherwise: ${D_t^{\sigma}} = (\abstrac \emptycon)\Omega$
  \end{varitemize} 
\end{varitemize}
\end{varitemize}
\end{proof}

It follows from Lemma~\ref{lemma:testcontext} that if two well-typed closed terms are context equivalent, 
they are bisimilar:
\begin{theorem}\label{theo:fullabstr}
  Let $\termone,\termtwo$ be terms such that $\emptyset \proves \termone \equiv \termtwo: \sigma$.
  Then $\emptyset \proves \termone \Ev \termtwo: \sigma$.
\end{theorem}
\begin{proof}
Let $\testone$ be a test. We have that, since $\termone \equiv \termtwo$,
$$
  \prsucc{\mccbv}{(\termone,\sigma)}{\testone}=\sum\diS{\act{C_{\testone}^\sigma}{\termone}}
   =\sum\diS{\act{C_{\testone}^\sigma}{\termtwo}} 
   =\prsucc{\mccbv}{(\termtwo,\sigma)}{\testone},
$$
where $C_{\testone}^\sigma$ is the context from Lemma~\ref{lemma:testcontext}.
By Theorem \ref{theo:testbisim}, $(\termone,\sigma)$ and $(\termtwo,\sigma)$ are bisimilar.
So $\emptyset \proves \termone \Ev \termtwo: \sigma$ which is the thesis.
\end{proof}
We can now easily extend this result to terms in $\termtyg \sigma \Gamma$, which gives us 
Full Abstraction: bisimilarity and context equivalence indeed coincide.
\begin{theorem}[Full Abstraction]\label{theorem:fullabstraction}
Let $\termone$ and $\termtwo$ be terms in $\termtyg \sigma \Gamma$ .Then $\Gamma \proves \termone \equiv \termtwo: \sigma$ iff
$\Gamma \proves \termone \Ev \termtwo: \sigma$.
\end{theorem}
\begin{proof}
There is only one inclusion to show (we know already that bisimilarity is included in context equivalence). 
We know that $\equiv$ is value substitutive. 
We note $\Gamma = x_1: \tau_1,\ldots,x_n:\tau_n$.
So for all $\valone_1 \in \val {\tau_1},\ldots,\valone_n\in \val{\tau_n}$, we have:
$\emptyset \proves M[\overline{\valone}/\overline{x}] \equiv N[\overline{\valone}/\overline{x}]: \sigma$.
By Theorem \ref{theo:fullabstr}, we have: $\emptyset \proves M[\overline{\valone}/\overline{x}]\backsim_v N[\overline{\valone}/\overline{x}]: \sigma$.
And so by definition of the open extension: 
$\Gamma \proves M \Ev N: \sigma$.
\end{proof}
\subsection{The Asymmetric Case}
Theorem~\ref{theorem:fullabstraction} establishes a precise correspondence between bisimulation and
context equivalence. This is definitely not the end of the story --- surprisingly enough, indeed,
\emph{simulation} and the contextual \emph{preorder} do not coincide, and this section gives a
counterexample, namely a pair of terms which can be compared in the context preorder but which are not
similar. 

Let us fix the following terms:
\begin{align*}
  \termone&=\abstr{\varone}{\abstr{\vartwo}{(\Omega\oplus I)}};\\
  \termtwo&=\abstr{\varone}{(\abstr{\vartwo}{\Omega})\oplus(\abstr{\vartwo}{I})}.
\end{align*}
Both these terms can be given the type $\sigma=\arr{\bty}{\arr{\bty}{\arr{\bty}{\bty}}}$ in
the empty context.
The first thing to note is that $\termone$ and $\termtwo$ 
cannot even be compared in the simulation preorder:
\begin{lemma}
  It is not the case that $\emptyset\proves\termone\Rv\termtwo:\sigma$ nor that $\emptyset\proves\termtwo\Rv\termone:\sigma$.
\end{lemma}
\begin{proof}
   The Markov Chain used to define $\precsim$ has the following form:
   
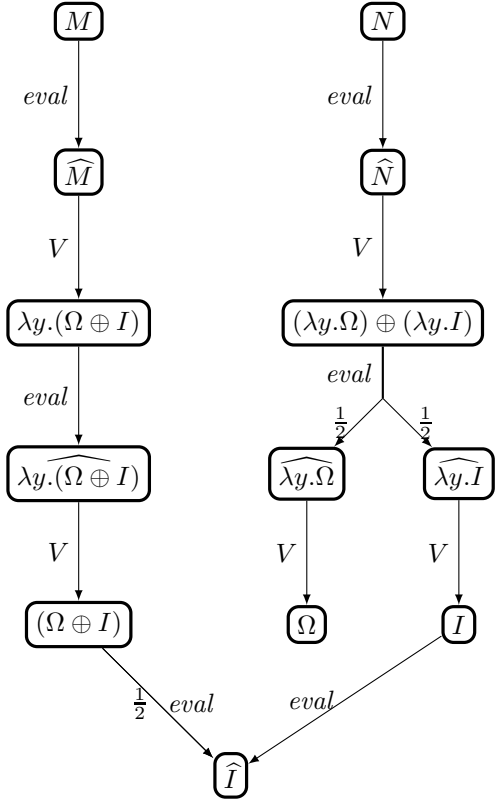
\begin{center}
     \begin{tikzpicture}[scale = 1]
       \tikzstyle{states}=[draw,rounded corners,very thick]
       \tikzstyle{operation}=[->,>=latex]
       \tikzstyle{trait}=[-,>=latex,thick]
       \tikzstyle{etiquette}=[midway]
       \tikzstyle{proba}=[midway]
       \node[states] (M) at (0,5){$\termone$};
       \node[states] (N) at (4,5){$\termtwo$};
       \node[states] (M2) at (0,3){ $\widehat \termone$};
       \node[states] (N2) at (4,3){ $\widehat \termtwo$};
       \node[states] (M3) at (0,1){ $\abstr{\vartwo}{(\Omega \oplus I)}$};
       \node[states] (N3) at (4,1){ $(\abstr {\vartwo}{\Omega})\oplus (\abstr {\vartwo}{I})$};
       \node[states] (M4) at (0,-1){ $\widehat {\abstr{\vartwo}{(\Omega \oplus I)}}$};
       \node[states] (N3a) at (3,-1){ $\widehat {\abstr {\vartwo}{\Omega}}$};
       \node[states] (N3b) at (5,-1){ $\widehat {\abstr {\vartwo}{I}}$};
       \node[states] (M5) at (0,-3){ ${(\Omega \oplus I)}$};
       \node[states] (N5a) at (3,-3){ ${\Omega}$};
       \node[states] (N5b) at (5,-3){ $I$};
       \node[states] (N6b) at (2,-5){ $\widehat I$};
       \coordinate (inter1) at (4,0);
       \draw[operation] (M)--(M2) node[etiquette, left]{ $\evact$};
       \draw[operation] (N)--(N2) node[etiquette, left]{ $\evact$};
       \draw[operation] (M2)--(M3) node[etiquette, left]{ $\valone$};
       \draw[operation] (N2)--(N3) node[etiquette, left]{ $\valone$};
       \draw[operation] (M3)--(M4) node[etiquette, left]{ $\evact$};
       \draw[trait] (N3)--(inter1) node[etiquette, left]{ $\evact$};
       {\draw[operation] (inter1)--(N3a) node[proba,left]{$\frac 1 2$};}
       {\draw[operation] (inter1)--(N3b) node[proba,right]{$\frac 1 2$};}
       \draw[operation] (M4)--(M5) node[etiquette, left]{ $\valone$};
       \draw[operation] (N3a)--(N5a) node[etiquette, left]{ $\valone$};
       \draw[operation] (N3b)--(N5b) node[etiquette, left]{ $\valone$};
       \draw[operation] (N5b)--(N6b) node[etiquette, left]{ $\evact$};
       {\draw[operation] (M5)--(N6b) node[proba,left]{$\frac 1 2$};}
       {\draw[operation] (M5)--(N6b) node[etiquette,right]{$\evact$};}
     \end{tikzpicture}
   \end{center}
  \begin{varitemize}
  \item
    Suppose that $\termtwo \precsim \termone$. So (since $\precsim$ is a simulation),  $\left(\abstrac I \oplus \abstrac \Omega \right) \precsim (\abstrac(I \oplus \Omega))$.
    So we have : $\frac 1 2 = \mccbvmatrix {(\left(\abstrac I \oplus \abstrac \Omega \right),\evact,\widehat{\abstrac I})} \leq \mccbvmatrix {(\left(\abstrac (I \oplus \Omega) \right),\evact,\precsim(\widehat{\abstrac I}))}$, and it folds that : $(\abstr \valone {(I \oplus \Omega)}) \in \precsim({\abstr \valone I})$, i.e. ($\abstr \valone I \precsim \abstr \valone (I \oplus \Omega)$). But since $\precsim$ is a simulation, we can then deduce that : $I \precsim I \oplus \Omega$ : and we have a contradiction since $\mccbvmatrix{(I,\evact, {\widehat I })} > \mccbvmatrix({I \oplus \Omega},{\evact},{\mccbvstates})$.
  \item 
    Suppose that $\termone \precsim \termtwo$. We use on the same way the fact that $\precsim$ is a simulation.
We have :  $(\abstrac(I \oplus \Omega)) \precsim \left(\abstrac I \oplus \abstrac \Omega \right)$. And so we have : $1 = \mccbvmatrix {(\left(\abstrac( I \oplus \Omega) \right),\evact,\widehat{\abstrac (I \oplus \Omega)})} \leq \mccbvmatrix {(\left(\abstrac I \oplus \abstrac \Omega \right),\evact,\precsim(\widehat{\abstrac (I \oplus \Omega)}))}$. It implies that :  $\widehat \abstrac (I \oplus \Omega) \precsim \widehat \abstrac \Omega $. 
    We can now apply the $\evact$ action, and we see that : $\mccbvmatrix({I \oplus \Omega},{\evact},{\mccbvstates}) \leq \mccbvmatrix ({I \oplus \Omega},{\evact},{\mccbvstates}) = 0$, and so we have a contradiction.
  \end{varitemize} 
  This concludes the proof.
\end{proof}
We now proceed by proving that $\termone$ and $\termtwo$ can be compared in the \emph{contextual} preorder. We will do
so by studying their dynamics seen as terms of $\Lambda_\oplus$~\cite{DalLagoZorzi}
(in which the only constructs are variables, abstractions, applications and probabilistic choices, and in which
types are absent) rather than terms of \PCFLP. We will later argue why this translates back into a result for \PCFLP.
This detour allows to simplify the overall treatment without sacrificing generality. From now on, then
$\termone$ and $\termtwo$ are seen as pure terms, where $\Omega$ takes the usual form $(\abstrac\varone\varone)(\abstrac\varone\varone)$.

Let us introduce some notation now. First of all, three terms need to be given names as follows:
$\termthree=\abstr{\vartwo}{(\Omega\oplus I)}$,
$\termthree_0=\abstr{\vartwo}{\Omega}$, and $\termthree_1=\abstr{\vartwo}{I}$.
If $b=b_1,\ldots,b_n\in\{0,1\}^n$, then $\termthree_b$ denotes the sequence of terms $\termthree_{b_1}\cdots\termthree_{b_n}$. 
If $\termfour$ is a term, $\evalvsp{\termfour}{\probone}$ means that there is distribution $\distrone$
such that $\evalvs{\termfour}{\distrone}$ and $\sum\distrone=\probone$ (where $\Rightarrow$ is small-step approximation
semantics~\cite{DalLagoZorzi}.

The idea, now, is to prove that in any term $\termfour$, if we replace an occurrence of $\termone$ by an occurrence
of $\termtwo$, we obtain a term $\termfive$ which converges with probability smaller than the one with which 
$\termfour$ converges. We first need an auxiliary lemma, which proves a similar result for $\termthree_0$ and $\termthree_1$.
\begin{lemma}\label{lemma:auxcounter}
  For every term $\termfour$, if $\evalvsp{(\subst{\termfour}{\termthree_0}{\varone})}{\probone}$, then
  there is another real number $\probtwo\geq\probone$ such that $\evalvsp{(\subst{\termfour}{\termthree_1}{\varone})}{\probtwo}$.
\end{lemma}
\begin{proof}
  First, we can remark that, for every term $\termfour$ and any variable $\varthree$ which doesn't appear in 
  $\termfour$, $\subst \termfour {\termthree_0} \varone = \subst {\left(\subst \termfour {\abstr{\vartwo}{\varthree}}{\varone}\right)}{\Omega}{\varthree}$, 
  and $\subst \termfour {\termthree_1} \varone = \subst {\left(\subst \termfour {\abstr{\vartwo}{\varthree}}{\varone}\right)}{I}{\varthree}$. 
  It is thus enough to show that for every term $\termfive$, if $\evalvsp{(\subst{\termfive}{\Omega}{\varone})}{\probone}$, then there is 
  $\probtwo \geq\probone$ such that $\evalvsp{(\subst{\termfive}{I}{\varone})}{\probtwo}$. 
  This is an induction on the proof of $\evalvsp{(\subst{\termfive}{\Omega}{\varone})}{\probone}$, i.e., an
  induction on the structure of a derivation of $\evalvs{(\subst{\termfive}{\Omega}{\varone})}{\distrone}$
  where $\sum{\distrone}=\probone$. Some interesting cases:
  \begin{varitemize}
  \item 
    If ${(\subst{\termfive}{\Omega}{\varone})} = \valone$ is a value, then the term 
    ${(\subst{\termfive}{I}{\varone})}$ is a value too. So we have 
    ${\evalvsD {(\subst{\termfive}{I}{\varone})}{\distrv{(\subst \termfive I \varone)}}}$, and so 
    $\evalvsp{(\subst{\termfive}{I}{\varone})}{1}$, and the thesis holds.
  \item 
    Suppose that the derivation looks as follows:
    \begin{prooftree}
      \AxiomC{${(\subst{\termfive}{\Omega}{\varone})} \rightarrow \overline{\termsix}$}
      \AxiomC{$ \evalvs {\termsix_i}{\distrtwo_i}$}
      \BinaryInfC{$\evalvsD {(\subst{\termfive}{\Omega}{\varone})}  {\sum_{1\leq i \leq k}{\frac {1} {k} } \cdot {\distrtwo_i}}$} 
    \end{prooftree}    
    Then there are two possible cases : 
    \begin{varitemize}  
    \item If $\subst{\termfive}{\Omega}{\varone}\rightarrow\termsix_1,\ldots,\termsix_k$, but the
      involved redex is \emph{not} $\Omega$, then we can easily
      prove that each $\termsix_i$ can be written in the
      form $\subst{\termseven_i}{\Omega}{\varone}$, where
      $$
      \subst{\termfive}{\Omega}{\varone}\rightarrow
      \subst{\termseven_1}{\Omega}{\varone},\ldots,
      \subst{\termseven_k}{\Omega}{\varone}.
      $$
      Similarly $
      \subst{\termfive}{I}{\varone}\rightarrow
      \subst{\termseven_1}{I}{\varone},\ldots,
      \subst{\termseven_k}{I}{\varone}$.
      We can then apply the induction hypothesis to each of the derivations for $\subst{\termseven_i}{\Omega}{\varone}$.
    \item 
      The interesting case is when the active redex in $\subst{\termfive}{\Omega}{\varone}$ is $\Omega$. Since we have 
      $\Omega \rightarrow \Omega$, we have $\subst{\termfive}{\Omega}{\varone} \rightarrow \subst {\termfive}{\Omega}{\varone}$, and so 
      $\overline{\termsix}=\termsix_1 = \subst {\termfive}{\Omega}{\varone}$, and $\distrone = \distrtwo_1$. 
      We can apply the induction hypothesis to $\evalvs {\termsix_1}{\distrtwo_1}$, and the thesis follows.
    \end{varitemize} 
  \end{varitemize}
  This concludes the proof.\cqed
\end{proof}
We are now ready to prove the central lemma of this section, which takes a rather complicated form just for the sake of
its inductive proof:
\begin{lemma}\label{lemma:bound}
  Suppose that $\termfour$ is a term and suppose that $\evalvsp{(\subst{\termfour}{\termone,\termthree}{\varone,\overline{\vartwo}})}{\probone}$,
  where $\overline{\vartwo}=\vartwo_1,\ldots,\vartwo_n$. Then for every $b\in\{0,1\}^n$ there is $\probone_b$ such
  that $\evalvsp{(\subst{\termfour}{\termtwo,\termthree_b}{\varone,\overline{\vartwo}})}{\probone_b}$
  and $\sum_{b}\frac{\probone_b}{2^n}\geq\probone$.
\end{lemma}
\begin{proof}
  This is an induction on the proof of $\evalvsp{(\subst{\termfour}{\termone,\termthree}{\varone,\overline{\vartwo}})}{\probone}$, i.e., an
  induction on the structure of a derivation of $\evalvs{(\subst{\termfour}{\termone,\termthree}{\varone,\overline{\vartwo}})}{\distrone}$
  where $\sum{\distrone}=\probone$:
  \begin{varitemize}
  \item
    If $\subst{\termfour}{\termone,\termthree}{\varone,\overline{\vartwo}}$ is a value, then:
    \begin{varitemize}
    \item
      either $\probone=1$, but we can also choose $\probone_b$ to be $1$ for every $b$, 
      since the term $\subst{\termfour}{\termtwo,\termthree_b}{\varone,\overline{\vartwo}}$ is a value, too;
    \item
      or $\probone=0$, and in this case we can fix $\probone_b$ to be $0$ for every $b$.
    \end{varitemize}
  \item
    If $\subst{\termfour}{\termone,\termthree}{\varone,\overline{\vartwo}}\rightarrow\termfive_1,\ldots,\termfive_k$, but the
    involved redex has \emph{not} $\termone$ nor $\termthree$ as functions, then we are done, because one can easily
    prove in this case that each $\termfive_i$ can be written in the
    form $\subst{\termsix_i}{\termone,\termthree}{\varone,\overline{\vartwo}}$, where
    $$
    \subst{\termfour}{\termtwo,\termthree_b}{\varone,\overline{\vartwo}}\rightarrow
    \subst{\termsix_1}{\termtwo,\termthree_b}{\varone,\overline{\vartwo}},\ldots,
    \subst{\termsix_k}{\termtwo,\termthree_b}{\varone,\overline{\vartwo}}.
    $$
    It suffices, then, to apply the induction hypothesis to each of the derivations
    for $\subst{\termsix_i}{\termone,\termthree}{\varone,\overline{\vartwo}}$, easily reaching the thesis;
  \item
    The interesting case is when the active redex in $\subst{\termfour}{\termone,\termthree}{\varone,\overline{\vartwo}}$
    has either $\termone$ or $\termthree$ (or, better, occurrences of them coming from the
    substitution) in functional position.
    \begin{varitemize}
    \item
      If $\termone$ is involved, then there are a term $\termfive$ and a variable $\varthree$ such that
      \begin{align*}
        \subst{\termfour}{\termone,\termthree}{\varone,\overline{\vartwo}}&\rightarrow\subst{\termfive}{\termone,\termthree,\termthree}{\varone,\overline{\vartwo},\varthree};\\
        \subst{\termfour}{\termtwo,\termthree_b}{\varone,\overline{\vartwo}}&
          \rightarrow\subst{\termfive}{\termtwo,\termthree_b,\termthree_0}{\varone,\overline{\vartwo},\varthree},
          \rightarrow\subst{\termfive}{\termtwo,\termthree_b,\termthree_1}{\varone,\overline{\vartwo},\varthree}.
      \end{align*}
      This, in particular, means that we can easily apply the induction hypothesis to $\subst{\termfive}{\termone,\termthree,\termthree}{\varone,\overline{\vartwo},\varthree}$.
    \item
      If, on the other hand $\termthree$ is involved in the redex, then there are a term $\termfive$ and a variable $\varthree$ such that
      $$
      \subst{\termfour}{\termone,\termthree}{\varone,\overline{\vartwo}}\rightarrow
      \subst{\termfive}{\termone,\termthree,\Omega}{\varone,\overline{\vartwo},\varthree},
      \subst{\termfive}{\termone,\termthree,I}{\varone,\overline{\vartwo},\varthree}.
      $$
      Moreover, the space of all sequences $b$ can be partitioned into two classes of the same cardinality $2^{n-1}$, call
      them $B_B$ and $B_G$; for every $b\in B_B$, we have that $\subst{\termfour}{\termtwo,\termthree_b}{\varone,\overline{\vartwo}}$ 
      is diverging, while for every $b\in B_G$, we have that
      $$
      \subst{\termfour}{\termtwo,\termthree_b}{\varone,\overline{\vartwo}}\rightarrow\subst{\termfive}{\termtwo,\termthree_b,I}{\varone,\overline{\vartwo},\varthree}.
      $$
      Observe how for any $b\in B_B$ there is $\hat{b}\in B_G$ such that $b$ and $\hat{b}$ agree on every bit except one, which is $0$ in $b$ and $1$ in $\hat{b}$.
      Now, observe that $\probone=\frac{\probtwo}{2}$ where $\evalvsp{\subst{\termfive}{\termone,\termthree,I}{\varone,\overline{\vartwo},\varthree}}{\probtwo}$.
      We can then apply the induction hypothesis and obtain that $\probtwo\leq\sum_{b}\frac{\probtwo_b}{2^n}$ where
      $\evalvsp{\subst{\termfive}{\termtwo,\termthree_b,I}{\varone,\overline{\vartwo},\varthree}}{\probtwo_b}$. Due to Lemma~\ref{lemma:auxcounter}, we
      can assume without losing generality that $\probtwo_b\leq\probtwo_{\hat{b}}$ for every $b\in B_B$. Now, fix $\probone_b=0$ if $b\in B_B$
      and $\probone_b=\probtwo_b$ if $b\in B_G$. Of course $\evalvsp{(\subst{\termfour}{\termtwo,\termthree_b}{\varone,\overline{\vartwo}})}{\probone_b}$.
      But moreover,
      $$
      \probone=\frac{\probtwo}{2}\leq \frac{1}{2}\sum_{b}\frac{\probtwo_b}{2^n}
      \leq \frac{1}{2}\sum_{b\in B_G}\frac{2\cdot\probtwo_b}{2^n}=\sum_{b\in B_G}\frac{\probtwo_b}{2^n}=\sum_{b}\frac{\probone_b}{2^n}.
      $$
    \end{varitemize}
  \end{varitemize}
This concludes the proof.\cqed
\end{proof}
From what we have seen so far, it is already clear that for any context $\ctxone$, it cannot be that $\sum\diS{\ctxone[\termone]}>\sum\diS{\ctxone[\termtwo]}$,
as this would mean that for a certain term $\termfour$, $\subst{\termfour}{\termone}{\varone}$ would converge to a distribution $\distrone$
whose sum $\probone$ is higher than the sum of any distribution to which $\subst{\termfour}{\termtwo}{\varone}$ converges, and this is in contradiction
with Lemma~\ref{lemma:bound}: simply consider the case where $n=0$. 

But how about \PCFLP? Actually, there is an embedding $\emb{\cdot}$ of \PCFLP\ into $\Lambda_{\oplus}$ such that for every
$\termfour\in\termty{\sigma}$, it holds that $\sum\diS{\termfour}=\sum\diS{\emb{\termfour}}$ (for more details, see the next section). 
As a consequence there cannot be any \PCFLP\ context contradicting what we have said in the last paragraph. Summing up,
\begin{theorem}
The simulation preorder $\Rv$ is not fully abstract.
\end{theorem}
The careful reader may now wonder whether a result akin to Theorem~\ref{theo:testbisim} exists for \emph{simulation}
and testing. Actually, there \emph{is} such a result~\cite{vBMOW05}, but for a different notion of test, which not only, like $\testlan$,
includes conjunctive tests, but also disjunctive ones. Now, anybody familiar with the historical developments of the quest 
for a fully abstract model of \PCF\ \cite{Plotkin77,BerryCurien82} would immediately recognize disjunctive tests as something 
which cannot be easily implemented by terms. 
\subsection{Embedding \PCFLP\ into $\Lambda_\oplus$}
The embedding $\emb{\cdot}$ maps any term in \PCFLP\ into a pure, untyped, term. It is defined as follows:
\begin{align*}
\emb{\varone}&=\varone;\\
\emb{\cnst{n}}&=\embsn{n};\\
\emb{\cnst{b}}&=\embsn{b};\\
\emb{\nil}&=\abstr{\varone}{\abstr{\vartwo}{\varone\star}};\\
\emb{\lst{\termone}{\termtwo}}&=\abstr{\varone}{\abstr{\vartwo}{\vartwo\emb{\termone}\emb{\termtwo}}};\\
\emb{\pair{\termone}{\termtwo}}&=\abstr{\varone}{\varone(\abstrsv{\emb{\termone}})(\abstrsv{\emb{\termtwo}})};\\
\emb{\abstr{\varone}{\termone}}&=\abstr{\varone}{\emb{\termone}};\\
\emb{\fix{\termone}}&=\abstr{\vartwo}{\termone_\fixt(\abstr{\varone}{\emb{\termone}})\vartwo};\\
\emb{\termone\oplus\termtwo}&=\emb{\termone}\oplus\emb{\termtwo};\\
\emb{\ifth{\termone}{\termtwo}{\termthree}}&=\emb{\termone}(\abstrsv{\emb{\termtwo}})(\abstrsv{\emb{\termthree}})\dmyval;\\
\emb{\op{\termone}{\termtwo}}&=\termone_{\op{}{}}\emb{\termone}\emb{\termtwo};\\
\emb{\fst{\termone}}&=\emb{\termone}(\abstr{\varone}{\abstr{\vartwo}{\varone}})\star;\\
\emb{\snd{\termone}}&=\emb{\termone}(\abstr{\varone}{\abstr{\vartwo}{\vartwo}})\star;\\
\emb{\termone\termtwo}&=\emb{\termone}\emb{\termtwo};\\
\emb{\caset{\termone}{\termtwo}{\termthree}}&=\emb{\termone}(\abstrsv{\emb{\termtwo}})(\abstr{h}{\abstr{t}{\termthree}});
\end{align*}
where:
\begin{varitemize}
\item
  $\embsn{\cdot}$ is the so-called Scott-encoding of natural numbers and booleans in the $\lambda$-calculus:
\item
  $\termone_\fixt$ is the term $\termtwo\termtwo$, where $\termtwo$ is the term
  $\abstr{\varone}{\abstr{\vartwo}{\vartwo\left(\abstr{\varthree}{((\varone\varone)\vartwo)\varthree}\right)}}$.
\item
  $\termone_{\op{}{}}$ is the term implementing $\op{}{}$, which we suppose to always exist given the universality
  of weak call-by-value reduction~\cite{DalLagoMartini12}.
\end{varitemize}
\begin{lemma}
  For every \PCFLP\ term $\termone$, $\termone$ is a value iff $\emb{\termone}$ is a value.
\end{lemma}
\begin{lemma}
  For every typable \PCFLP\ term $\termone$, if $\evalvs{\termone}{\distrone}$, then $\evalvs{\emb{\termone}}{{\emb{\distrone}}}$.
\end{lemma}
\begin{proposition}
  For every typable \PCFLP\ term $\termone$, $\emb{\diS{\termone}}=\diS{\emb{\termone}}$.
\end{proposition}
\section{A Comparison with Call-by-Name}
Actually, \PCFLP\ could easily be endowed with call-by-name rather than call-by-value operational semantics. The obtained calculus, then,
is amenable to a treatment similar to the one described in Section~\ref{sect:appbis}. Full abstraction, however, holds neither for simulation
nor for bisimulation. These results are anyway among the major contributions of~\cite{ADLS13}.
The precise correspondence between testing and bisimulation described in Section~\ref{sect:fromlmptolmc} shed some further light on
the gap between call-by-value and call-by-name evaluation. In both cases, indeed, bisimulation can be characterized by testing as
given in Definition~\ref{def:testlan}. What call-by-name evaluation misses, however, is the capability to copy a term \emph{after} having
evaluated it, a feature which is instead available if parameters are passed to function evaluated, as in 
call-by-value. In a sense, then, the tests corresponding to bisimilarity are
the same in call-by-name, but the calculus turns out to be too poor to implement all of them. We conjecture that
the subclass of tests which are implementable in a call-by-name setting are those in the form $\seq{\testone_1,\ldots,\testone_n}$
(where each $\testone_i$ is in the form $\actone_i^1\cdot\ldots\cdot\actone_i^{m_i}\cdot\omega$), and that full abstraction
can be recovered if the language is endnowed with an operator for \emph{sequencing}.

\section{Conclusions}
In this paper, we study probabilistic applicative bisimulation in a call-by-value scenario, in the meantime generalizing it
to a typed language akin to Plotkin's \PCF. Actually, some of the obtained results turn out to be surprising, highlighting a gap
between the symmetric and asymmetric cases, and between call-by-value and call-by-name evaluation. This is a phenomenon
which simply does not show up when applicative bisimulation is defined over deterministic~\cite{Abramsky90} nor over 
nondeterministic~\cite{Las98a} $\lambda$-calculi. The path towards these results goes through a characterization of 
bisimilarity by testing which is known from the literature~\cite{vBMOW05}.
Noticeably, the latter helps in finding the right place for probabilistic $\lambda$-calculi in the coinductive spectrum: the corresponding
notion of test is more powerful than plain trace equivalence, but definitely less complex than
the infinitary notion of test which characterizes applicative bisimulation in presence of nondeterminism~\cite{Ong93}.

Further work includes a broader study on (not necessarily coinductive) notions of equivalence for probabilistic $\lambda$-calculi. As an
example, it would be nice to understand the relations between applicative bisimulation and logical relations (e.g. the ones defined
in~\cite{GLLN08}). Another interesting direction would be the study of notions of \emph{approximate} equivalence for $\lambda$-calculi with
restricted expressive power. This would be a step forward getting a coinductive characterization of computational indistinguishability,
with possibly nice applications for cryptographic protocol verification.

\bibliographystyle{abbrv}
\bibliography{biblio}
\end{document}